\documentclass[11pt,letterpaper]{article}

\newif\ifdraft\draftfalse

\newif\iflong\longtrue

\usepackage[utf8]{inputenc}
\usepackage[english]{babel}
\usepackage{authblk,amssymb,amsfonts,amsthm,amsmath}
\usepackage{xspace}
\usepackage{hyperref}
\usepackage{verbatim}
\usepackage{listings}
\usepackage[caps]{complexity}

\usepackage[compact]{titlesec}
\usepackage{times,paralist}

\usepackage[margin=1in]{geometry}
\let\olditem\item
\renewcommand{\item}{\setlength{\itemsep}{0pt}\setlength{\parskip}{0pt}\setlength{\parsep}{0pt}\setlength{\topsep}{0pt}\olditem}

\usepackage[svgnames,dvipsnames]{xcolor}
\usepackage{tikz}
\usetikzlibrary{arrows,decorations.pathreplacing,backgrounds}
\usepackage{algorithmic}
\usepackage{algorithm}
\usepackage{todonotes}
\usepackage{microtype}
\usepackage{caption}
\usepackage{subcaption}

\newtheorem{theorem}{Theorem}[section]
\newtheorem{definition}[theorem]{Definition}
\newtheorem{lemma}[theorem]{Lemma}
\newtheorem{fact}[theorem]{Fact}

\newtheorem{proposition}[theorem]{Proposition}

\theoremstyle{plain}

\newcommand{\eps}{\varepsilon}
\newcommand{\size}[1]{\lvert #1 \rvert}

\newcommand{\dist}{\mathsf{dist}}
\newcommand{\bdist}{\mathsf{bdist}}

\newcommand{\final}{\mathtt{final}}
\newcommand{\temp}{\mathtt{temp}}
\newcommand{\low}{\mathtt{low}}
\newcommand{\high}{\mathtt{high}}

\newcommand{\w}{\mathcal{W}}

\newcommand{\push}{\mathsf{push}}
\newcommand{\pop}{\mathsf{pop}}

\newcommand{\Order}{\mathrm{O}}

\makeatletter

\floatname{algorithm}{Algorithm}

\newcommand{\Sc}{\Sigma_{\text{\tt +}}}
\newcommand{\Sr}{\Sigma_{\text{\tt -}}}
\newcommand{\Sl}{\Sigma_{\text{\tt =}}}

\newcommand{\Sinf}{\SQ} 
\newcommand{\SQ}{\Sigma_{Q}}
\newcommand{\LQ}{\Lambda_{Q}}
\newcommand{\La}{\Lambda}
\newcommand{\Linf}{\LQ} 
\renewcommand{\S}{\Sigma}
\newcommand{\zug}[1]{\langle #1 \rangle}

\newcommand{\Qin}{Q_\mathit{in}}
\newcommand{\Qf}{Q_\mathit{f}}

\newcommand{\VPL}{{{\sc Vpl}}}
\newcommand{\VPA}{{\sc Vpa}\xspace}

\newcommand{\earrow}[1]{\mbox{$\stackrel{#1}{\longrightarrow}$}}

\newcommand{\vpa}{\mathcal{A}}

\newcommand{\vstack}{\Gamma}

\newcommand{\sh}{\mathrm{height}}

\newcommand{\Next}{\mathrm{Next}}
\newcommand{\Depth}{\mathrm{Depth}}
\newcommand{\Prefix}{\mathrm{Prefix}}
\newcommand{\Top}{\mathrm{top}}

\title{Streaming Property Testing of Visibly Pushdown Languages\footnote{Partially supported by the French ANR projects ANR-12-BS02-005 (RDAM) and ANR-14-CE25-0017 (AGREG)}}
\author[1]{Nathana\"el Fran\c{c}ois\thanks{\texttt{nathanael.francois@tu-dortmund.de}}}
\author[2]{Fr\'ed\'eric Magniez\thanks{\texttt{frederic.magniez@cnrs.fr}}}
\author[3]{Michel de Rougemont\thanks{\texttt{mdr@liafa.univ-paris-diderot.fr}}}
\author[2]{Olivier Serre\thanks{\texttt{Olivier.Serre@cnrs.fr}}}
\affil[1]{Fakultät für Informatik, TU Dortmund, Germany}
\affil[2]{CNRS, LIAFA, Univ Paris Diderot, Sorbonne Paris-Cit\'e, France}
\affil[3]{University of Paris II and LIAFA, CNRS, France}
\date{}

\makeatletter
\def\timenow{\@tempcnta\time
  \@tempcntb\@tempcnta
  \divide\@tempcntb60
  \ifnum10>\@tempcntb0\fi\number\@tempcntb
  \multiply\@tempcntb60
  \advance\@tempcnta-\@tempcntb
  :\ifnum10>\@tempcnta0\fi\number\@tempcnta}
\makeatother

\ifdraft

\newcommand\os[1]{\todo[inline,size=\scriptsize,backgroundcolor=PaleTurquoise,caption={}]{#1 - \textbf{Olivier}}}
\newcommand\nf[1]{\todo[inline,size=\scriptsize,backgroundcolor=Yellow,caption={}]{#1 - \textbf{Nathana\"el}}}
\newcommand\mdr[1]{\todo[inline,size=\scriptsize,backgroundcolor=SpringGreen,caption={}]{#1 - \textbf{Michel}}}
\newcommand\fm[1]{\todo[inline,size=\scriptsize,backgroundcolor=pink,caption={}]{#1 - \textbf{Fr\'ed\'eric}}}

\else

\newcommand\os[1]{}
\newcommand\nf[1]{}
\newcommand\mdr[1]{}
\newcommand\fm[1]{}

\fi

\iflong
\ifdraft
\newcommand\vlong[1]{{\color{Brown} #1}}
\else
\newcommand\vlong[1]{#1} 
\fi
\newcommand\vshort[1]{} 
\makeatletter
\let\lst@floatdefault\lst@float
\makeatother
\else
\usepackage[nomarkers,figuresonly,nofiglist]{endfloat}
\makeatletter
\let\lst@floatdefault\lst@float
\makeatother

\newcommand\vlong[1]{} 
\newcommand\vshort[1]{#1} 
\fi

\begin{document}
\maketitle

\begin{abstract}
In the context of language recognition, we demonstrate the superiority of streaming property testers against streaming algorithms and property testers, when they are not combined.
Initiated by Feigenbaum {\it et al.},
a streaming property tester is a streaming algorithm recognizing a language 
under the property testing approximation: it must distinguish inputs of the language from those that are $\eps$-far from it, while using the smallest possible memory (rather than limiting its number of input queries). 

Our main result is a streaming $\eps$-property tester for visibly pushdown languages (\VPL{}) with one-sided error using memory space $\mathrm{poly}((\log n) / \eps)$.

This constructions relies on a  (non-streaming) property tester for weighted regular languages based on a previous tester by Alon {\it et al}. 
We provide a simple application of this tester for streaming testing special cases of instances of \VPL{}
that are already hard for both streaming algorithms and property testers.

Our main algorithm is a combination of an original simulation of visibly pushdown automata
using a stack with small height but possible items of linear size. In a second step,
those items are replaced by small sketches.
Those sketches relies on a notion of suffix-sampling we introduce. This sampling
is the key idea connecting our streaming tester algorithm to property testers.
\end{abstract}

\newpage


\section{Introduction}

Visibly pushdown languages (\VPL) play an important role in formal languages
with crucial applications for databases and program analysis.
In the context of structured documents, 
they are closely related with regular languages of unranked trees as captured by hedge automata.
A well-known result \cite{Alur07} states that, when the tree is given by its depth-first traversal, such automata 
correspond to visibly pushdown automata (\VPA) 
(see~\emph{e.g.} \cite{Libkin06} for an overview on automata and logic for unranked trees).
In databases, this word encoding of trees is known as XML encoding,
where DTD specifications are examples of often considered subclasses of {\VPL}.
In program analysis, {\VPA} also capture natural properties of execution traces of recursive finite-state programs, including non-regular ones such as those with pre and post conditions as expressed in the temporal logic of calls and returns (CaRet)~\cite{AlurEM04,AlurABEIL07}. 

Historically, {\VPL} got several names such as input-driven languages or, more recently, languages of nested words. 
Intuitively, a {\VPA} is a pushdown automaton whose actions on stack (push, pop or nothing)
are solely decided by the currently read symbol. As a consequence, symbols can be partitioned into three groups:
push, pop and neutral symbols.
The complexity of {\VPL} recognition 
has been addressed  in various computational models.
The first results go back to the design of logarithmic space algorithms~\cite{BraunmuhlV83} as well as
NC\textsuperscript{1}-circuits~\cite{Dymond88}.
Later on, other models motivated by the context of massive data were considered,
such as streaming algorithms and property testers (described below).

Streaming algorithms (see~\emph{e.g.} \cite{M02}) have only a sequential access to their input, on which
they can perform a single pass, or sometimes a small number of additional passes.
The size of their internal (random access) memory is the crucial complexity parameter,
which should be sublinear in the input size, and even polylogarithmic if possible.
The area of streaming algorithms has experienced tremendous growth
in many applications since the late 1990s. The analysis of Internet traffic~\cite{ams99}, in
which traffic logs are queried, was one of their first applications.
Nowadays, they have found applications with big data, notably to test graphs properties, and more recently in language recognition on very large inputs. 
The streaming complexity of language recognition has been firstly
considered for languages that arise in the context of memory checking~\cite{Blum95,ckm07},
of databases~\cite{sv02,ss07}, and later on for formal languages~\cite{MMN10,Babu10}.
However, even for simple {\VPL}, any randomized streaming algorithm with $p$ passes
requires memory $\Omega(n/p)$, where $n$ is the input size~\cite{km13}.

As opposed to streaming algorithms, (standard) property testers~\cite{bk95,blr93,GGR96} have random access
to their input but in the query model. They must query each piece of the input they need to access.
They should sample only a sublinear fraction of their input, and ideally make a constant number of queries.
In order to make  the task of verification possible, decision problems need to be approximated as follows.
Given a distance on words, an $\eps$-tester for a language $L$ distinguishes with high probability the words in $L$
from those $\eps$-far from $L$, using as few queries as possible. 
Property testing of regular languages was first considered for the Hamming distance~\cite{AKNS00}.
When the distance allows sufficient modifications of the input, such as moves of arbitrarily large factors,
it has been shown that any context-free language becomes testable with a constant number of 
queries~\cite{MR04,FMR10}.
However, for more realistic distances, property testers for simple languages require a large number of queries, 
especially if they have one-sided error only. For example the complexity of an $\eps$-tester for well-parenthesized expressions with two types of parentheses is between $\Omega(n^{1/11})$ and 
$\Order(n^{2/3})$~\cite{prr03}, and it becomes linear, even for one type of parentheses, if we require one-sided error~\cite{AKNS00}. The difficulty of testing regular tree languages was also addressed when the tester can directly query the tree structure~\cite{NdioneLN13,n15}.

Faced by the intrinsic hardness of {\VPL} in both streaming and property testing, we study the complexity
of {\em streaming property testers} of formal languages, a model of algorithms combining both approaches.
Such testers were historically introduced for testing specific problems  (groupedness)~\cite{fksv02} relevant for network data. They were
 later  studied in the context of testing the insert/extract-sequence of a priority-queue structure~\cite{ckm07}. We extend these studies to classes of problems.
A streaming property tester is a streaming algorithm recognizing a language 
under the property testing approximation: it must distinguish inputs of the language from those that are $\eps$-far from it, 
while using the smallest possible memory (rather than limiting its number of input queries). 
Such an algorithm can simulate any standard non-adaptive property tester.
Moreover, we will see that, using its full scan of the input, it can construct better sketches than in the query model. 

In this paper, we consider  a natural notion of distance for {\VPL}, the {\em balanced-edit distance}, which 
refines the edit distance on {\em balanced words} (where for each push symbol there is a matching pop symbol
at the same height of the stack, and conversly).
It can be interpreted as the edit distance on trees when trees are encoded as balanced words. Neutral symbols can be deleted/inserted, but any push symbol can only be deleted/inserted together with its matching pop symbol. 
Since our distance is larger than the standard edit distance, our testers are also valid for that distance.

In Section~\ref{exact}, we first design an exact algorithm that maintains a small stack but whose items can be of linear size
as opposed to the standard simulation of a pushdown automaton which usually has a stack of possible linear size but with constant size items.
In our algorithm, stack items are prefixes of some peaks (which we call unfinished peaks), where
a {\em peak} is a balanced factor whose push symbols appear all before the first pop symbol.
Our algorithm compresses an unfinished peak $u=u_+ v_-$ when it is followed by a long enough sequence.
More precisely, the compression applies to the peak $v_+ v_-$ obtained by disregarding part of the prefix of push sequence $u_+$. Those peaks are then inductively replaced, and therefore compressed, 
by the state-transition relation they define on the given automaton.
The relation is then considered as a single symbol whose weight is the size of the peak it represents.
In addition, to maintain a stack of logarithmic depth, 
one of the crucial properties of our algorithm (\textbf{Proposition~\ref{smallstack}}) 
is rewriting the input word as a peak formed by potentially a linear number of intermediate peaks, but with 
only a logarithmic number of nested peaks.

In Section~\ref{sec:simple}, for the case of a single peak,
we show how to sketch the current unfinished peak of our algorithm.
The simplicity of those instances will let us highlight our first idea. Moreover, they are already expressive enough 
in order to demonstrate the superiority of streaming testers against streaming algorithms and property testers, 
when they are not combined.
We first reduce the problem of streaming testing such instances to the problem of testing regular languages
in the standard model of property testing (\textbf{Theorem~\ref{finalLtester}}).
Since our reduction induces weights on the letters of the new input word, 
we need a tester for weighted regular languages\vlong{~(\textbf{Theorem~\ref{TheoRLtester}})}. 
Such a property tester has previously been devised in~\cite{n15} extending  constructions for unweighted regular languages~\cite{AKNS00,NdioneLN13}.
However, we consider a slightly simpler construction that could be of independent interest.
As a consequence we get a streaming property tester with polylogarithmic memory for recognizing peak instances of any given {\VPL} (\textbf{Theorem~\ref{TheoAlgoMountain}}), a task already hard for streaming algorithms and property testers (\textbf{Fact~\ref{hard}}).

In Section~\ref{sec:general}, we construct our main tester for a {\VPL} $L$ given by some  \VPA{}.
For this we introduce a more involved notion of sketches
made of a polylogarithmic number of samples. 
They are based on a new notion of suffix sampling  (\textbf{Definition~\ref{suffixsampling}}). 
This sampling consists in a decomposition of the string into an increasing sequence of suffixes, whose weights increase geometrically. Such a decomposition can be computed online on a data stream, 
and one can maintain samples in each suffix of the decomposition using a standard reservoir sampling. This suffix decomposition will allow us to simulate an appropriate sampling on the peaks we compress, even if we do not yet know where they start.
Our sampling can be used to perform an approximate computation of the compressed relation
by our new property tester of weighted regular languages which we also used for single peaks.
We first establish a result of stability which basically states that we can assume that our algorithm knows in advance where the peak it will compress starts (\textbf{Lemma~\ref{stability}}).  
Then we prove the robustness of our algorithm: words that are $\eps$-far from $L$ are rejected with high probability (\textbf{Lemma~\ref{robustness}}). 
As a consequence, we get a one-pass streaming $\eps$-tester for $L$ with one-sided error $\eta$ 
and memory space $\Order(m^52^{3m^2}(\log n)^6 (\log 1/\eta)/ \eps^4)$, where $m$ is the number of states of a \VPA{} recognizing $L$  (\textbf{Theorem~\ref{TheoMain}}).


\section{Definitions and Preliminaries}\label{sec:prelim}
Let $\mathbb{N}^*$ be the set of positive integers, and for any integer $n\in\mathbb{N}^*$, let $[n]=\{1,2,\ldots,n\}$.
A $t$-subset of a set $S$ is any subset of $S$ of size $t$.
For a finite alphabet $\S$ we denote the set of finite words over $\S$ by
$\S^*$. 
For a word $u=u(1)u(2)\cdots u(n)$, we call $n$ the \emph{length} of $u$, and
$u(i)$ the $i$th letter in $u$. We write $u[i,j]$ for the factor $u(i)u(i+1)\cdots u(j)$ of $u$.
When we mention letters and factors of $u$
we implicitly also mention their positions in $u$.
We say that $v$ is a \emph{sub-factor} of $v'$, denoted $v\leq v'$,
if $v=u[i,j]$ and $v'=u[i',j']$ with $[i,j]\subseteq [i',j']$.
Similarly we say that $v=v'$ if $[i,j]=[i',j']$. If $i \leq i' \leq j \leq j'$ we say that the \emph{overlap}
of $v$ and $v'$ is $u[i',j]$. If $v$ is a sub-factor of $v'$ then the overlap of $v$ and $v'$ is $v$.
Given two multisets of factors $S$ and $S'$, we say that $S\leq S'$ if for each factor $v\in S$ there is a corresponding factor $v'\in S'$ such that $v\leq v'$.

\paragraph{Weighted Words and Sampling.}\label{pre-sampling}
A \emph{weight function} on a word $u$ with $n$ letters is a function
$\lambda:[n]\to\mathbb{N}^*$ on the letters of $u$,
whose value $\lambda(i)$ is called the \emph{weight of $u(i)$}. 
A \emph{weighted word} over $\S$ is a pair $(u,\lambda)$ where $u\in \S^*$ and $\lambda$ is a weight function on $u$.  We define $|u(i)|=\lambda(i)$ and $|u[i,j]|=\lambda(i)+\lambda(i+1)+\ldots+\lambda(j)$. 
The length of $(u,\lambda)$ is the length of $u$.
For simplicity, we will denote by $u$ the weighted word $(u,\lambda)$.
Weighted letters will be used to substitute factors of same weights. Therefore, restrictions may exist
on available weights for a given letter.

\looseness=-1
Our algorithms will be based on a sampling of small factors according to their weights.
We introduce a very specific notion adapted to our setting.
For a weighted word $u$, we denote by {\em $k$-factor sampling on $u$} 
the sampling over factors $u[i,i+l]$ with probability $|u(i)|/|u|$, where $l\geq 0$ is the smallest 
integer such that $|u[i,i+l]|\geq k$ if it exists, otherwise $l$ is such that $i+l$ is the last letter of $u$.
More generally, we call $k$-factor such a factor.
For the special case of $k=1$, we call this sampling a {\em letter sampling on $u$}.
Observe that both of them can be implemented using a standard reservoir sampling\vlong{~(see Algorithm~\ref{Reservoir} for letter sampling)}.

\iflong
\ifdraft \color{red} \fi
\begin{lstlisting}[caption={Reservoir Sampling},label=Reservoir,captionpos=t,abovecaptionskip=-\medskipamount,mathescape]
$\text{\bf Input:}$ Data stream $u$, Integer parameter $t>1$
$\text{\bf Data structure:}$ 
  $\sigma \gets 0$ // Current weight of the processed stream
  $S \gets$ empty multiset // Multiset of sampled letters 
$\text{\bf Code:}$ 
$i\gets 1$, $a\gets\Next(u)$, $\sigma\gets |a|$
$S\gets$ $t$ copies of $a$
While $u$ not finished
   $i++$, $a\gets\Next(u)$, $\sigma\gets \sigma+|a|$
   For each $b \in S$ 
      Replace $b$ by $a$ with probability $|a|/\sigma$
Output $S$
\end{lstlisting}
\ifdraft \color{black} \fi
\fi

Even if our algorithm will require several samples from a $k$-factor sampling,
we will often only be able to simulate this sampling by sampling either larger factors, more factors, or both.
Let $\mathcal{W}_1$ be a sampler producing a random multiset $S_1$ of factors of some given weighted word $u$.
Then $\mathcal{W}_2$ {\em over-samples} $\mathcal{W}_1$ if it produces a random multiset $S_2$ of factors of $u$ such that
for each factor $v$, we have 
$\Pr(\exists v' \in S_2\text{ such that $v$ is a factor of $v'$}) \geq \Pr (\exists v' \in S_1\text{ such that $v$ is a factor of $v'$})$.

\paragraph{Finite State Automata and Visibly Pushdown Automata.}
A \emph{finite state automaton} is a tuple of the form $\mathcal{A} =
(Q,\S,\Qin,\Qf,\Delta)$ where $Q$ is a finite set of control states,
$\S$ is a finite input alphabet, $\Qin\subseteq Q$ is a subset of initial
states, $\Qf\subseteq Q$ is a subset of final states and $\Delta\subseteq
Q\times \S \times Q$ is a transition relation. 
We write $p\earrow{u} q$, to mean that there is a sequence of transitions in $\mathcal{A}$
from $p$ to $q$ while processing $u$, and we call $(p,q)$ a {\em $u$-transitions}.
A word $u$ is accepted if $q_{in}\earrow{u} q_{f}$ for some $q_{in}\in \Qin$ and $q_f\in \Qf$. The language $L(\mathcal{A})$ of $\mathcal{A}$ is the set of words accepted by $\mathcal{A}$, and we refer to such a language as a  \emph{regular language}.
{For $\S'\subseteq\S$,
the \emph{$\S'$-diameter} (or simply \emph{diameter} when $\S'=\S$) 
of $\mathcal{A}$ is the maximum over all possible pairs $(p,q)\in Q^2$
of $\min\{|u| : p\earrow{u} q \text{ and } u\in\S'^*\}$, 
whenever this minimum is not over an empty set.
We say that $\mathcal{A}$ is \emph{$\S'$-closed}, when $p\earrow{u}q$ for some $u\in\S^*$
if and only if $p\earrow{u'}q$ for some $u'\in\S'^*$.}

A \emph{pushdown alphabet} is a triple
$\zug{\Sc, \Sr, \Sl}$ that comprises three disjoint finite alphabets: $\Sc$ is a
finite set of \emph{push symbols}, $\Sr$ is a finite set of
\emph{pop symbols}, and $\Sl$ is a finite set of \emph{neutral symbols}.  For any such triple, let $\S = \Sc \cup \Sr \cup
\Sl$.
Intuitively, a \emph{visibly pushdown automaton}~\cite{AM04} over $\zug{\Sc, \Sr, \Sl}$ is a pushdown automaton
restricted so that it pushes onto the stack only on reading a
push, it pops the stack only on reading a pop, and it does not
modify the stack on reading a neutral symbol.
Up to coding, this notion is similar to the one of input driven pushdown automata~\cite{mel80}
and of nested word automata~\cite{AM09}.

\begin{definition}[Visibly pushdown automaton~\cite{AM04}]
A \emph{visibly pushdown automaton} (\VPA) over $\zug{\Sc, \Sr,
\Sl}$ is a tuple $\vpa = (Q,\S,\vstack,\Qin,\Qf,\Delta)$ where $Q$ is a finite set
of states, $\Qin \subseteq Q$ is a set of initial states, $\Qf\subseteq
Q$ is a set of final states, $\vstack$ is a finite stack alphabet, 
and $\Delta\subseteq (Q \times \Sc \times Q \times \vstack) \cup (Q \times \Sr
\times \vstack \times Q) \cup (Q \times \Sl \times Q)$ is the
transition relation. 
\end{definition}

To represent stacks we use a special bottom-of-stack symbol $\bot$
that is not in $\vstack$.  
A \emph{configuration} of a {\VPA} {$\vpa$} is a pair $(\sigma, q)$,
where $q \in Q$ and $\sigma \in \bot \cdot \Gamma^*$. For $a\in\S$, there is an
\emph{$a$-transition} from a configuration $(\sigma,q)$ to
$(\sigma',q')$, denoted $(\sigma,q) \earrow{a} (\sigma',q')$,
in the following cases:
\begin{compactitem}
\item If $a$ is a push symbol, then $\sigma' = \sigma\gamma$ for some
  $(q, a, q', \gamma) \in \Delta$, and we write $q \earrow{a} (q', \push(\gamma))$.
\item If $a$ is a pop symbol, then $\sigma = \sigma'\gamma$ for some
  $(q, a, \gamma, q') \in \Delta$,  and we write $(q,\pop(\gamma)) \earrow{a} q'$.
\item If $a$ is a neutral symbol, then $\sigma = \sigma'$ and $(q,  a, q') \in \Delta$,
and we write $q \earrow{a} q'$.
\end{compactitem}
For a finite word $u = a_1\cdots a_n\in \S^{*}$, 
if $(\sigma_{i-1},q_{i-1}) \earrow{a_{i}} (\sigma_{i},q_{i})$ for every $1\leq i\leq n$, we also write
$(\sigma_0,q_0)\earrow{u}(\sigma_{n},q_{n})$.  
The word $u$ is \emph{accepted} by a {\VPA} if 
there is $(p,q)\in\Qin\times \Qf$ such that $(\bot,p)\earrow{u}(\bot,q)$.
The language $L(\vpa)$ of 
$\vpa$ is the set of words accepted by $\vpa$, and we refer to such a language as a \emph{visibly pushdown language} (\VPL).

At each step, the height of the stack is pre-determined by the prefix of
$u$ read so far.  
The \emph{height} $\sh(u)$ of $u\in\S^*$ is the difference between the number of its push symbols and of its pop symbols.
A word $u$ is \emph{balanced} if $\sh(u)=0$ and $\sh(u[1,i])\geq 0$ for all $i$.
We also say that a push symbol $u(i)$ \emph{matches} a pop symbol $u(j)$
if $\sh(u[i,j])=0$ and $\sh(u[i,k])> 0$ for all $i<k<j$. 
By extension, the height of $u(i)$ is $\sh(u[1,i-1])$ when $u(i)$ is a push symbol, and $\sh(u[1,i])$ otherwise.

For all balanced words $u$, the property $(\sigma,p)\earrow{u} (\sigma,q)$ does not depend on $\sigma$, therefore 
we simply write $p\earrow{u} q$, and say that $(p,q)$ is a \emph{$u$-transition}.
{We also define similarly to finite automata
the \emph{$\Sigma'$-diameter} of $\mathcal{A}$ (or simply diameter) and the notion $\mathcal{A}$ being \emph{$\Sigma'$-closed} on balanced words only.}


Our model is inherently restricted to input words having
no prefix of negative stack height, and
 we defined acceptance with an empty stack. This implies
that only balanced words can be accepted. From now on, we will always assume that the input is balanced as verifying this in a streaming context is easy.

\label{pt}

\paragraph{Balanced/Standard Edit Distance.}
The usual distance between words in property testing is the Hamming distance.
In this work, we consider an easier distance to manipulate in property testing
but still relevant for most applications, which is the edit distance, 
that we adapt to weighted words.  

Given a word $u$, we define two
possible \emph{edit operations}: 
the \emph{deletion} of a letter in position $i$ with corresponding cost
$|u(i)|$, and its converse operation, the \emph{insertion}
where we also select a weight, compatible with the restrictions on $\lambda$, for the new $u(i)$. 
Then the \emph{(standard) edit distance} $\dist(u,v)$
between two weighted words $u$ and $v$ is simply defined as the minimum
total cost of a sequence of edit operations changing $u$ to $v$. Note
that all letters that have not been inserted nor deleted must keep the
same weight. {For a restricted set of letters $\S'$, we also
define $\dist_{\S'}(u,v)$ where the insertions are restricted to letters in $\S'$.}

We will also consider a restricted version of this distance for balanced words, 
motivated by our study of \VPL{}.
Similarly, \emph{balanced-edit operations} can be deletions or insertions of letters,
but each deletion of a push symbol (resp. pop symbol) requires the deletion of the matching
pop symbol (resp. push symbol). Similarly for insertions: if a push
(resp. pop) symbol is inserted, then a matching pop (resp. push)
symbol must also be inserted simultaneously. The cost of these
operations is the weight of the affected letters, as with the edit
operations. 
We define the \emph{balanced-edit distance} $\bdist(u,v)$ between two
balanced words as the total cost of a sequence of balanced-edit operations changing $u$ to $v$. Similarly to $\dist_{\S'}(u,v)$ we define $\bdist_{\S'}(u,v)$.

When dealing with a visibly pushdown language, we will always use the balanced-edit distance,
whereas we will use the standard-edit distance for regular languages.
{We also say that $u$ is \emph{$(\eps,\S')$-far} from $v$
if $\dist_{\S'}(u,v) > \eps |u|$, or $\bdist_{\S'}(u,v) > \eps |u|$, depending on the context;
otherwise we say that $u$ is \emph{$(\eps,\S')$-close} to $v$.
We omit $\S'$ when $\S'=\S$.}

\paragraph{Streaming Property Testers.}

An $\eps$-tester for a language $L$ accepts all
inputs which belong to $L$ with probability $1$ and rejects with high
probability all inputs which are $\eps$-far from $L$, \emph{i.e.} that are $\eps$-far from any element of $L$. 
In particular, a tester for some given distance is also a tester for any other smaller distance.
Two-sided error testers have also been studied but in this paper 
we stay with the notion of one-sided testers, that we adapt in the context of streaming algorithm as in~\cite{fksv02}.
%

\begin{definition}[Streaming property tester]
Let $\eps > 0$ and let $L$ be a language.
A {\em streaming $\eps$-tester} for $L$ with one-sided error $\eta$ and memory $s(n)$ is a randomized
algorithm $A$ such that, for any input $u$ of length $n$ given as a data stream:
\begin{compactitem}
\item If $u\in L$, then $A$ accepts with probability $1$;
\item If $u$ is $\eps$-far from $L$, then $A$ rejects with probability at least $1-\eta$;
\item  $A$ processes $u$ within a single sequential pass while maintaining a memory space of $O(s(n))$ bits.
\end{compactitem}
\end{definition}


\section{Exact Algorithm}\label{exact}
Fix a \VPA{} $\vpa{}$ recognizing some \VPL{} $L$ on $\Sigma=\Sc\cup\Sr\cup\Sl$. In this section, we design an exact streaming algorithm that decides whether an input  belongs to $L$.
Algorithm~\ref{AlgoExact} maintains a stack of small height
but whose items can be of linear size.
In Section~\ref{sec:general}, we replace stack items by appropriated small sketches 

\subsection{Notations and Algorithm Description}
Call a \emph{peak} a sequence of push symbols  followed by a sequence of pop symbols,
with possibly intermediate neutral symbols,
\emph{i.e.}  an element of the language 
{$\Lambda=\bigcup_{j\geq 0}((\Sl)^*\cdot \Sc)^j\cdot(\Sl)^*\cdot (\Sr\cdot(\Sl)^*)^j.$
One can compress any pick $v\in\Lambda$ by the set $R_v=\{(p,q):p\earrow{v} q\}$ of the $v$-transitions,
and consider $R_v$ as a new neutral symbol with weight $|v|$.
In fact, for the purpose of the analysis of our algorithm, we augment neutral symbols by 
many more relations for which $\mathcal{A}$ remains $\Sigma$-closed.
For the rest of the paper,   they will be the only symbols with weight potentially larger than $1$.
\begin{definition}\label{sigmaq}
Let $\SQ$ be $\Sl$ augmented by all  weighted letters encoding a relation $R\subseteq Q\times Q$ such that for every $(p,q)\in R$ there is a balanced word $u\in\S^*$ with $p\earrow{u} q$.
Let  $\LQ$ be  $\Lambda$ where $\Sl$ is replaced by $\SQ$.
\end{definition}
We then write $p\earrow{R} q$ whenever $(p,q)\in R$, and extend $\mathcal{A}$ and $L$ accordingly.
Of course, our notion of distance will be solely based on the initial alphabet $\Sigma$.

A general balanced input instance $u$ will consist of many nested peaks.
However, we will recursively replace each factor $v\in\Lambda_Q$ by
$R_v$ with weight $|v|$.

Denote by $\Prefix(\Lambda_Q)$  the language of prefixes of words in $\Linf$.  
While processing the prefix $u[1,i]$ of the data stream $u$,
Algorithm~\ref{AlgoExact} maintains a suffix $u_0\in\Prefix(\Linf)$ of $u[1,i]$,
that is an unfinished peak, with some simplifications of factors $v$ in $\Linf$ by their corresponding relation $R_v$. 
Therefore $u_0$ consists of a sequence of push symbols and neutral symbols possibly followed by a sequence of pop symbols and neutral symbols.  The algorithm also maintains a subset $R_\temp\subseteq Q\times Q$ that is the set of transitions for the maximal prefix of $u[1,i]$ in $\Linf$. When the stream is over, the set $R_\temp$ is used to decide whether $u\in L$ or not.

When a push symbol $a$ comes after a pop sequence, $u_0\cdot a$ is no longer in $\Prefix(\Linf)$ 
hence,  Algorithm~\ref{AlgoExact} puts $u_0$ on the stack of unfinished peaks (see lines \ref{AE_R0} to \ref{C_R0} and Figure~\ref{fig:algoexact1}) and $u_0$ is reset to $a$. 
In other situations, it adds $a$ to $u_0$.
In case $u_0$ becomes a word in $\Linf$ (see lines \ref{AE_R1b} to \ref{C_R1} and Figure~\ref{fig:algoexact2}), Algorithm~\ref{AlgoExact} computes the
set of $u_0$-transitions $R_{u_0} \in \Sinf$, and adds $R_{u_0}$ to the previous unfinished peak that is retrieved on top of the stack and becomes the current unfinished peak; in the special case where the stack is empty one simply updates the set $R_\temp$ by taking its composition with $R_{u_0}$.


\begin{lstlisting}[caption={Exact Tester for a VPL},label=AlgoExact,captionpos=t,float,abovecaptionskip=-\medskipamount,mathescape]%pour la vlongue
%\begin{lstlisting}[caption={Exact Tester for a VPL},label=AlgoExact,captionpos=t,abovecaptionskip=-\medskipamount,mathescape]%pour la vcourte

$\text{\bf Input:}$ Balanced data stream $u$  
$\text{\bf Data structure:}$ 
  $Stack \gets$ empty stack  // Stack of items $v$ with $v\in\Prefix(\LQ)$
  $u_0 \gets \emptyset$  // $u_0\in\Prefix(\LQ)$ is a suffix of the processed part $u[1,i]$ of $u$
       // with possibly some factors $v\in\LQ$ replaced by $R_v$
  $R_\temp\gets \{(p,p)\}_{p\in Q}$ // Set of transitions for the maximal prefix of $u[1,i]$ in $\LQ$
$\text{\bf Code:}$
While $u$ not finished
  $a \gets \Next(u)$ //Read and process a new symbol $a$
  If $a\in\Sc$ and $u_0$ has a letter in $\Sr$  // $u_0\cdot a \not \in \Prefix(\LQ)$ (*@ \label{AE_R0} @*)
    Push $u_0$ on $Stack$, $u_0 \gets a$(*@ \label{C_R0} @*)
  Else $u_0 \gets u_0 \cdot a$ 
  If $u_0$ is balanced  // $u_0 \in \LQ$: compression  (*@ \label{AE_R1b} @*)
    Compute $R_{u_0}$ the set of $u_0$-transitions (*@ \label{AE_R1} @*)
    If $Stack=\emptyset$, then $R_\temp\gets R_\temp\circ R_{u_0}$, $u_0 \gets \emptyset$ 
       // where $\circ$ denotes the composition of relations
    Else Pop $v$ from $Stack$, $u_0 \gets v \cdot R_{u_0}$ (*@ \label{C_R1} @*)
  Let $(v_1\cdot v_2) \gets \Top(Stack)$ s.t. $v_2$ is maximal and balanced  // $v_2\in\Linf$ (*@\label{AE_v2}@*)
  If $|u_0| \geq |v_2|/2$  (*@\label{AE_test_left}@*)  // $u_0$ is big enough and $v_2$ can be replaced by $R_{v_2}$
    Compute $R_{v_2}$ the set of $v_2$-transitions(*@\label{AE_R2}@*), Pop $v$ from Stack, $u_0 \gets (v_1\cdot R_{v_2}) \cdot u_0$  (*@ \label{C_R2} @*)
If $(\Qin \times \Qf) \cap R_\temp \not = \emptyset$, Accept; Else Reject // $R_\temp=R_u$ (*@\label{AE_final_check}@*) 
\end{lstlisting}

\subsection{Algorithm Analysis}
We now introduce the quantity $\Depth(v)$ for each factor $v$ constructed in Algorithm~\ref{AlgoExact}. It quantifies the number of processed nested picks in $v$ as follows:
\begin{definition}
For each factor constructed in Algorithm~\ref{AlgoExact}, $\Depth$ is defined dynamically  by $\Depth(a)=0$ when $a\in\S$, $\Depth(v)=\max_i\Depth(v(i))$ and $\Depth(R_{v})=\Depth(v)+1$.
\end{definition}

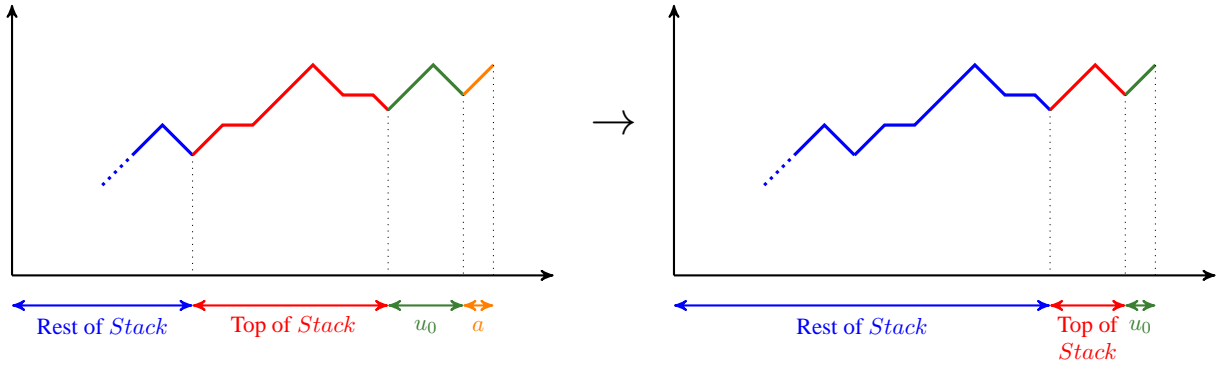
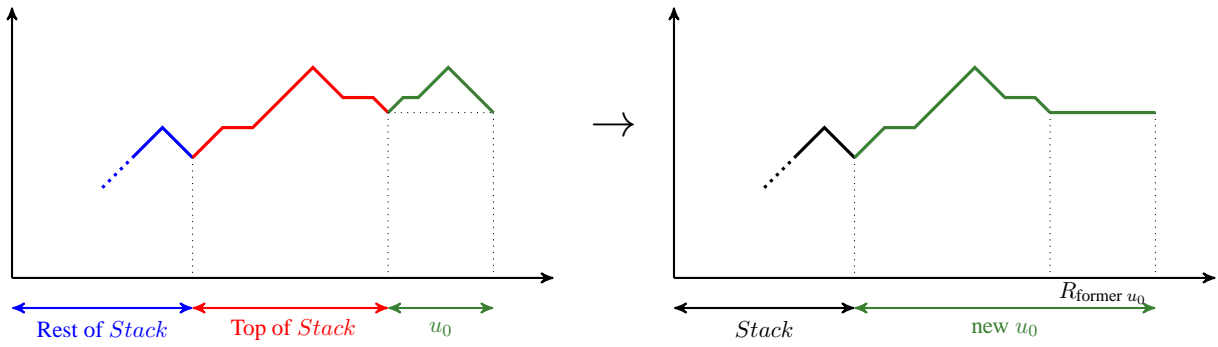
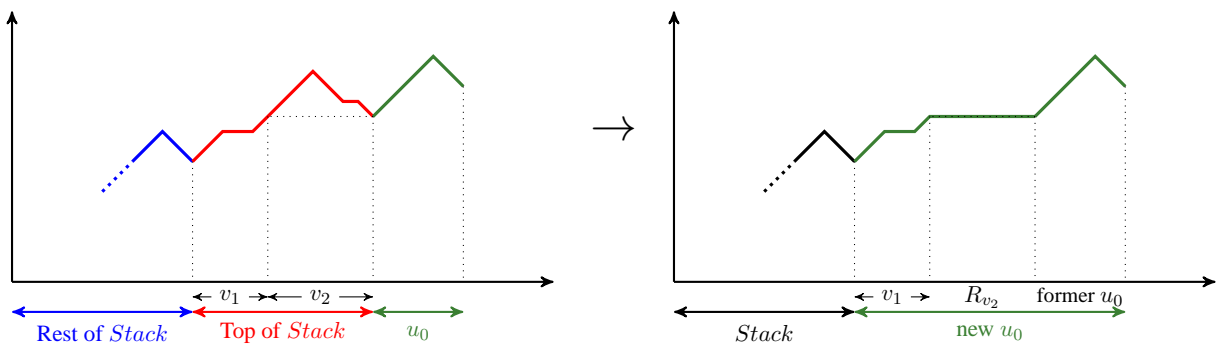
\begin{figure}
\begin{subfigure}[b]{\linewidth}
\begin{tikzpicture}[scale=0.8,transform shape,
stack/.style={very thick,blue},
topstack/.style={very thick,red},
u/.style={very thick,OliveGreen},
a/.style={very thick,orange},
stackline/.style={thick,blue},
topstackline/.style={thick,red},
uline/.style={thick,OliveGreen},
aline/.style={thick,orange},
]
\begin{scope}
\draw[thick,->,>=stealth'] (-3,-2) -- (6,-2);\draw[thick,->,>=stealth'] (-3,-2) -- (-3,2.5);
\draw [stack,dotted] (-1.5,-.5) -- (-1,0);
\draw [stack] (-1,0) -- (-0.5,.5) -- (0,0);
\draw[<->,>=stealth',stackline] (-3,-2.5) -- (0,-2.5);\node[stackline] at (-1.5,-2.85) {Rest of $Stack$};
\draw [topstack] (0,0) -- (.5,.5) -- (1,.5) -- (2,1.5) -- (2.5,1)--(3,1)--(3.25,0.75);
\draw[<->,>=stealth',topstackline] (0,-2.5) -- (3.25,-2.5);\node[topstackline] at (1.675,-2.85) {Top of $Stack$};
\draw [u] (3.25,.75) -- (4,1.5) -- (4.5,1);
\draw[<->,>=stealth',uline] (3.25,-2.5) -- (4.5,-2.5);\node[uline] at (3.875,-2.85) {$u_0$};
\draw [a] (4.5,1) -- (5,1.5);
\draw[<->,>=stealth',aline] (4.5,-2.5) -- (5,-2.5);\node[aline] at (4.75,-2.85) {$a$};
\draw[dotted] (0,0) -- (0,-2);
\draw[dotted] (3.25,.75) -- (3.25,-2);
\draw[dotted] (4.5,1) -- (4.5,-2);
\draw[dotted] (5,1.5) -- (5,-2);
\end{scope}

\begin{scope}[xshift = 11cm]
\node[scale=2] at (-4,.5) {$\rightarrow$};
\draw[thick,->,>=stealth'] (-3,-2) -- (6,-2);\draw[thick,->,>=stealth'] (-3,-2) -- (-3,2.5);
\draw [stack,dotted] (-1.5,-.5) -- (-1,0);
\draw [stack] (-1,0) -- (-0.5,.5) -- (0,0);
\draw[<->,>=stealth',stackline] (-3,-2.5) -- (3.25,-2.5);\node[stackline] at (.125,-2.85) {Rest of $Stack$};
\draw [stack] (0,0) -- (.5,.5) -- (1,.5) -- (2,1.5) -- (2.5,1)--(3,1)--(3.25,0.75);
\draw [topstack] (3.25,.75) -- (4,1.5) -- (4.5,1);
\draw[<->,>=stealth',topstackline] (3.25,-2.5) -- (4.5,-2.5);\node[topstackline] at (3.875,-2.85) {Top of};\node[topstackline] at (3.875,-3.25) {$Stack$};
\draw [u] (4.5,1) -- (5,1.5);
\draw[<->,>=stealth',uline] (4.5,-2.5) -- (5,-2.5);\node[uline] at (4.75,-2.85) {$u_0$};
\draw[dotted] (3.25,.75) -- (3.25,-2);
\draw[dotted] (4.5,1) -- (4.5,-2);
\draw[dotted] (5,1.5) -- (5,-2);
\end{scope}
\end{tikzpicture}
\caption{Illustration of lines \ref{AE_R0} to \ref{C_R0} from Algorithm~\ref{AlgoExact}}\label{fig:algoexact1}
\end{subfigure}
\medskip

\begin{subfigure}[b]{\linewidth}
\begin{tikzpicture}[scale=0.8,transform shape,
stack/.style={very thick,blue},
topstack/.style={very thick,red},
u/.style={very thick,OliveGreen},
a/.style={very thick,orange},
stackline/.style={thick,blue},
topstackline/.style={thick,red},
uline/.style={thick,OliveGreen},
aline/.style={thick,orange},
]
\begin{scope}
\draw[thick,->,>=stealth'] (-3,-2) -- (6,-2);\draw[thick,->,>=stealth'] (-3,-2) -- (-3,2.5);
\draw [stack,dotted] (-1.5,-.5) -- (-1,0);
\draw [stack] (-1,0) -- (-0.5,.5) -- (0,0);
\draw[<->,>=stealth',stackline] (-3,-2.5) -- (0,-2.5);\node[stackline] at (-1.5,-2.85) {Rest of $Stack$};
\draw [topstack] (0,0) -- (.5,.5) -- (1,.5) -- (2,1.5) -- (2.5,1)--(3,1)--(3.25,0.75);
\draw[<->,>=stealth',topstackline] (0,-2.5) -- (3.25,-2.5);\node[topstackline] at (1.675,-2.85) {Top of $Stack$};
\draw [u] (3.25,.75) -- (3.5,1) -- (3.75,1) -- (4.25,1.5) -- (4.5,1.25) -- (5,.75);
\draw[<->,>=stealth',uline] (3.25,-2.5) -- (5,-2.5);\node[uline] at (4.125,-2.85) {$u_0$};
\draw[dotted] (0,0) -- (0,-2);
\draw[dotted] (3.25,.75) -- (5,.75);
\draw[dotted] (3.25,.75) -- (3.25,-2);
\draw[dotted] (5,.75) -- (5,-2);
\end{scope}

\begin{scope}[xshift = 11cm]
\node[scale=2] at (-4,.5) {$\rightarrow$};
\draw[thick,->,>=stealth'] (-3,-2) -- (6,-2);\draw[thick,->,>=stealth'] (-3,-2) -- (-3,2.5);
\draw [stack,dotted,black] (-1.5,-.5) -- (-1,0);
\draw [stack,black] (-1,0) -- (-0.5,.5) -- (0,0);
\draw[<->,>=stealth',stackline,black] (-3,-2.5) -- (0,-2.5);\node[stackline,black] at (-1.5,-2.85) {$Stack$};
\draw [u] (0,0) -- (.5,.5) -- (1,.5) -- (2,1.5) -- (2.5,1)--(3,1)--(3.25,0.75);
\draw [u] (3.25,.75)  -- (5,.75);
\draw[<->,>=stealth',uline] (0,-2.5) -- (5,-2.5);\node[uline] at (2.5,-2.85) {new $u_0$};
\node at (4.125,-2.25) {$R_{\text{former } u_0}$};
\draw[dotted] (0,0) -- (0,-2);
\draw[dotted] (3.25,.75) -- (3.25,-2);
\draw[dotted] (5,.75) -- (5,-2);
\end{scope}
\end{tikzpicture}
\caption{Illustration of lines \ref{AE_R1b} to \ref{C_R1} from Algorithm~\ref{AlgoExact}}\label{fig:algoexact2}
\end{subfigure}
\medskip

\begin{subfigure}[b]{\linewidth}
\begin{tikzpicture}[scale=0.8,transform shape,
stack/.style={very thick,blue},
topstack/.style={very thick,red},
u/.style={very thick,OliveGreen},
a/.style={very thick,orange},
stackline/.style={thick,blue},
topstackline/.style={thick,red},
uline/.style={thick,OliveGreen},
aline/.style={thick,orange},
]
\begin{scope}
\draw[thick,->,>=stealth'] (-3,-2) -- (6,-2);\draw[thick,->,>=stealth'] (-3,-2) -- (-3,2.5);
\draw [stack,dotted] (-1.5,-.5) -- (-1,0);
\draw [stack] (-1,0) -- (-0.5,.5) -- (0,0);
\draw[<->,>=stealth',stackline] (-3,-2.5) -- (0,-2.5);\node[stackline] at (-1.5,-2.85) {Rest of $Stack$};
\draw [topstack] (0,0) -- (.5,.5) -- (1,.5) -- (2,1.5) -- (2.5,1)--(2.75,1)--(3,0.75);
\draw[<->,>=stealth',topstackline] (0,-2.5) -- (3,-2.5);\node[topstackline] at (1.5,-2.85) {Top of $Stack$};
\draw [u] (3,.75) -- (4,1.75) -- (4.5,1.25);
\draw[<->,>=stealth',uline] (3,-2.5) -- (4.5,-2.5);\node[uline] at (3.75,-2.85) {$u_0$};
\draw[dotted] (0,0) -- (0,-2);
\draw[dotted] (3,.75) -- (3,-2);
\draw[dotted] (3,.75) -- (1.25,.75);
\draw[dotted] (1.25,.75) -- (1.25,-2);
\draw[dotted] (4.5,1.25) -- (4.5,-2);
\node(v2) at (2.125,-2.25) {$v_2$};\draw[<-,>=stealth'] (1.25,-2.25) -- (v2);\draw[<-,>=stealth'] (3,-2.25) -- (v2);
\node(v1) at (.625,-2.25) {$v_1$};\draw[<-,>=stealth'] (1.25,-2.25) -- (v1);\draw[<-,>=stealth'] (0,-2.25) -- (v1);
\end{scope}

\begin{scope}[xshift = 11cm]
\node[scale=2] at (-4,.5) {$\rightarrow$};
\draw[thick,->,>=stealth'] (-3,-2) -- (6,-2);\draw[thick,->,>=stealth'] (-3,-2) -- (-3,2.5);
\draw [stack,dotted,black] (-1.5,-.5) -- (-1,0);
\draw [stack,black] (-1,0) -- (-0.5,.5) -- (0,0);
\draw[<->,>=stealth',stackline,black] (-3,-2.5) -- (0,-2.5);\node[stackline,black] at (-1.5,-2.85) {$Stack$};
\draw [u] (0,0) -- (.5,.5) -- (1,.5) -- (1.25,0.75)  --(3,0.75);
\draw [u] (3,.75) -- (4,1.75) -- (4.5,1.25);
\draw[<->,>=stealth',uline] (0,-2.5) -- (4.5,-2.5);\node[uline] at (2.25,-2.85) {new $u_0$};
\draw[dotted] (0,0) -- (0,-2);
\draw[dotted] (3,.75) -- (3,-2);
\draw[dotted] (3,.75) -- (1.25,.75);
\draw[dotted] (1.25,.75) -- (1.25,-2);
\draw[dotted] (4.5,1.25) -- (4.5,-2);
\node(v2) at (2.125,-2.25) {$R_{v_2}$};
\node(v1) at (.625,-2.25) {$v_1$};\draw[<-,>=stealth'] (1.25,-2.25) -- (v1);\draw[<-,>=stealth'] (0,-2.25) -- (v1);
\node(former) at (3.75,-2.25){{\small former} $u_0$};
\end{scope}

\end{tikzpicture}
\caption{Illustration of lines \ref{AE_v2} to \ref{C_R2} from Algorithm~\ref{AlgoExact}}\label{fig:algoexact3}
\end{subfigure}
\caption{Illustration of Algorithm~\ref{AlgoExact}}
\end{figure}

In order to bound the size of the stack,  Algorithm~\ref{AlgoExact} considers the maximal balanced suffix $v_2$ of the topmost element $v_1\cdot v_2$ of the stack and, whenever $|u_0|\geq |v_2|/2$,
it computes the relation $R_{v_2}$ and continues with a bigger current peak starting with $v_1$ (see lines \ref{AE_v2} to \ref{C_R2} and Figure~\ref{fig:algoexact3}).
A consequence of this compression is that the elements in the stack have geometrically decreasing weight and therefore the height of the stack used by Algorithm~\ref{AlgoExact} is logarithmic in the length of the input stream. This can be proved by a direct inspection of Algorithm~\ref{AlgoExact}.

\begin{proposition}\label{smallstack}
Algorithm~\ref{AlgoExact} accepts exactly when $u\in L$, while maintaining a stack of at most  $\log |u|$ items.
\end{proposition}

%

We state that Algorithm~\ref{AlgoExact}, when processing an input $u$ of length $n$,  considers at most $\Order(\log n)$ nested picks, that is $\Depth(v)=\Order(\log n)$ for all factors constructed in Algorithm~\ref{AlgoExact}.

\begin{lemma}
\label{NestedR}
Let $v$ be the factor used to compute $R_v$ at line either~\ref{AE_R1} or~\ref{AE_R2} of
Algorithm~\ref{AlgoExact}.
Then $|v(i)|\leq 2|v|/3$, for all $i$.
Moreover, for any factor $w$ constructed by Algorithm~\ref{AlgoExact} it holds that $\Depth(w)=\Order(\log |w|)$.
\end{lemma}
\vlong{
\begin{proof}
One only has to consider letters in $\Sinf$. Hence, let $R_w$ belongs to $v$ for some $w$: either $w$ was simplified into $R_w$ at line~\ref{AE_R1} or at line~\ref{AE_R2} of Algorithm~\ref{AlgoExact}.

Let us first assume that it was done at line~\ref{AE_R2}. Therefore, 
there is some $v' \in \Prefix(\Linf)$ to the right of $w$ with total
weight greater than $|w|/2=\size{R_w}/2$. This factor $v'$
is entirely contained within $v$: indeed, when $R_w$ is computed $v$ includes $v'$. 
Therefore $\size{R_w} \leq 2\size{v}/3$.

If $R_w$ comes from line~\ref{AE_R1}, then $w=u_0$ and this $u_0$ is balanced and compressed. We claim that at the previous round the test in line~\ref{AE_test_left} failed, that is $|u_0|-1\leq |v_2|/2$ where $v_2$ is the maximal balanced suffix of $\mathrm{top}(Stack)$. Indeed, when performing the sequence of actions following a positive test in line~\ref{AE_test_left}, the number of unmatched push symbols in the new $u_0$ is augmented at least by $1$ from the previous $u_0$: hence, it cannot be equal to $1$ as the elements in the stack have pending call symbols and therefore in the next round $u_0$ cannot be balanced. Therefore one has $|u_0|-1\leq |v_2|/2$. Now when $R_w= R_{u_0}$ is created, it is contains in a factor that also contains $v_2$ and at least one pending call before $v_2$. Hence, $\size{R_w}\leq  2\size{v}/3$.

Finally, the fact that for any factor $w$ constructed by Algorithm~\ref{AlgoExact}, $\Depth(w)=\Order(\log |w|)$ derives from the fact that if $\Depth(w) =k$, then $|w| \geq (3/2)^k$. This can in turn be shown by induction on the depth. Obviously any factor will have weight at least $1$. Let us assume all factors of depth $k$ have weight at least $(3/2)^k$, and let $w(i)$ be a letter such that $\Depth(w(i)) = k+1$. By definition, $w(i) = R_v$ for some factor $v$ with $\Depth(v) =k$. This means $v$ contains at least one letter $v(j)$ of depth $k$. By our induction hypothesis, $|v(j)| \geq  (3/2)^k$, and therefore $|w(i)| = |v| \geq (3/2)|v(j)| \geq (3/2)^{k+1}$. 
\end{proof}
}


\section{The Special Case Of Peaks}
\label{sec:simple}


We now consider restricted instances consisting of a \emph{single  peak}.
For these instances, Algorithm~\ref{AlgoExact} never uses its stack but $u_0$ can be of linear size.
We show how to replace $u_0$ by a small random sketch
in order to get a streaming property tester using polylogarithmic memory.
In Section~\ref{sec:general},
this notion of sketch will be later extended to obtain our final streaming property tester for general instances.

\subsection{Hard Peak Instances}
Peaks are already hard for both streaming algorithms and property testing algorithms. Indeed, consider the language $\mathrm{Disj}\subseteq \Lambda$ over alphabet $\S=\{0,1,\overline{0},\overline{1},a\}$ and defined as the union of all languages
$a^*\cdot x(1)\cdot a^* \cdot \ldots \cdot x(j) \cdot a^* \cdot\overline{y(j)} \cdot a^*\cdot \ldots \cdot \overline{y(1)} \cdot a^*$, where $j\geq 1$, $x,y\in\{0,1\}^j$, and $x(i)y(i)\not=1$ for all $i$.

Then $\mathrm{Disj}$ can be recognized by a \VPA{} with $3$ states,
 $\Sc=\{0,1\}$, $\Sr = \{\overline{0},\overline{1}\}$ and $\Sl =\{a\}$.
However, the following fact states its hardness for both models. The hardness for non-approximation streaming algorithms comes for a standard reduction to Set-Disjointness. The hardness for property testing algorithms is a corollary of a similar result due to~\cite{prr03} for parenthesis languages with two types of parentheses. 

\begin{fact}\label{hard}
Any randomized $p$-pass streaming algorithm for $\mathrm{Disj}$ 
requires memory space $\Omega(n/p)$, where $n$ is the input length.
Moreover, any (non-streaming) $(2^{-6})$-tester for $\mathrm{Disj}$ requires to query $\Omega(n^{1/11}/\log n)$ letters of the input word.
\end{fact}
\vlong{\begin{proof}
The Set-Disjointness problem is defined as follows. Two players have respectively a $A$ and $B$ of $\{1,\dots,n\}$ and they must output whether $A \cap B = \emptyset$. The communication complexity of this problem is well known to be $\Omega(n)$. Therefore using the standard reduction of streaming algorithms to communication protocols, any randomised $p$-pass algorithm for $\mathrm{Disj}$ will require memory space $\Omega(n/p)$.

To prove the hardness of testing $\mathrm{Disj}$ in the query model, we use a result from~\cite{prr03} (Theorem 2) which states that any Hamming distance query model property tester for $\lang{Par}_2 \cap \La$ the language on the alphabet $\{(,[,],),*\}$ consisting of well-parenthesized words that are also in $\La$ requires $\Omega(n^{1/11})$ queries.

We first note that because of the way $\lang{Par}_2 \cap \La$ is constructed, the Hamming distance and the edit distance of any word in $\La$ to $\lang{Par}_2 \cap \La$ are within a constant factor of one another. Indeed, if a sequence of insertions brings some word $u$ inside $\lang{Par}_2 \cap \La$, then the deletions of the parentheses matching the insertions would do the same. And all deletions can similarly be replaced by a substitution of the character being deleted with $*$.

It is also easy to reduce that language to $\mathrm{Disj}$: we replace $($ by $01$, $)$ by $\overline{0}\overline{1}$, $[$ by $10$, and $]$ by $\overline{1}\overline{0}$.
\end{proof}
}

Surprisingly, for every $\eps>0$, we will show that  languages of the form $L\cap \Lambda$, where $L$ is a \VPL{}, become easy to $\eps$-test by streaming algorithms. This is mainly because, given their full access to the input, streaming algorithms can perform an input sampling which makes the property testing task easy, using only a single pass and few memory.

\subsection{Slicing Automaton}\label{slice_automaton}
Observe that Algorithm~\ref{AlgoExact} will never use the stack in the case of a single peak. After Algorithm~\ref{AlgoExact} has processed the $i$-th letter of the data stream, $u_0$ contains $u[1,i]$.
We will show how to compute $R_{u_0}$ at line~\ref{AE_R1} using a standard finite state automaton without any stack.

Indeed, for every \VPL{} $L$, one can construct a regular language $\widehat{L}$ such that
testing whether $u \in L\cap \Lambda$ is equivalent to test whether 
some other word $\widehat{u}$ belongs to $\widehat{L}$.
For this, let $\mathrm{I}$ be a special symbol not in $\Sl$ encoding the relation set $\{(p,p):p\in Q\}$.
For a word $v\in\Sl^l$, write $[v,\mathrm{I}]$ for the word $(v(1),\mathrm{I})\cdot(v(2),\mathrm{I})\cdots(v(l),\mathrm{I})$, and similarly $[\mathrm{I},v]$.
Consider a weighted word of the form
$u = \Big(\prod_{i=1}^j v_i\cdot a_i\Big) \cdot v_{j+1} \cdot \Big(\prod_{i=j}^1 {b_{i}}\cdot w_i\Big),$
where $a_i\in\Sc$, ${b_i}\in\Sr$, and $v_i,w_i\in\Sl^*$.
Then the \emph{slicing} of $u$ (see Figure~\ref{fig:mountain}) is the word $\widehat{u}$ over the alphabet $\widehat{\S} = (\Sc \times \Sr)\cup (\Sl\times \{\mathrm{I}\}) \cup  (\{\mathrm{I}\}\times\Sl)$ 
defined by
$\widehat{u}= \Big(\prod_{i=1}^j [v_i,\mathrm{I}] \cdot [\mathrm{I},w_i] \cdot (a_i,{b_i})\Big)\cdot [v_{j+1},\mathrm{I}].$


\begin{definition} 
Let $\vpa=(Q,\S,\vstack,\Qin,\Qf,\Delta)$ be a \VPA.
The \emph{slicing} of $\mathcal{A}$ is the finite automaton \linebreak
$\widehat{\vpa}=(\widehat{Q},\widehat{\S},\widehat{\Qin},\widehat{\Qf},\widehat{\Delta})$ 
where
$\widehat{Q} = Q \times Q$,
$\widehat{Q_{in}} = Q_{in} \times \Qf$,
$\widehat{\Qf} =\{(p,p) : p \in Q\}$, and the transitions $\widehat{\Delta}$ are:
\begin{compactenum}
\item $(p,q) \earrow{(a,b)} (p',q')$ when $p \earrow{a} (p', \push(\gamma))$ and 
$(q', \pop(\gamma)) \earrow{b} q$ are both transitions of $\Delta$.
\item $(p,q) \earrow{(c,\mathrm{I})} (p',q)$,
resp. $(p,q) \earrow{(\mathrm{I},c)} (p,q')$,
when $p \earrow{c} p'$, resp. $q \earrow{c} q'$, is a transition of $\Delta$.
\end{compactenum}
\end{definition}

\begin{figure}[t]
\begin{center}
\begin{tikzpicture}[transform shape,scale=0.8]
\path[draw] (0,0) -- (5,5) -- (10,0);
\node at (-0.5,-.5) {$u=v_1$};
\node at (10,-.54) {$w_1$};
\node at (0.7,-.5) {$a_1$};
\node at (1.4,-.5) {$\cdots$};
\node at (2.5,-.5) {$a_i$};
\node at (2,-.5) {$v_i$};
\node at (3,-.5) {$v_{i+1}$};
\node at (3.8,-.5) {$\cdots$};
\node at (4.5,-.5) {$a_h$};
\node at (5,-.5) {\small $v_{h+1}$};
\node at (9.3,-.5) {$b_1 $};
\node at (8.7,-.5) {$\cdots$};
\node at (7.5,-.5) {$b_i$};
\node at (8,-.5) {$w_i$};
\node at (7,-.5) {$w_{i+1}$};
\node at (6.1,-.5) {$\cdots$};
\node at (5.5,-.5) {$b_h$};
\node at (0,0) {$\bullet$};\node at (1,1) {$\bullet$};\node at (2,2) {$\bullet$};\node at (3,3) {$\bullet$};\node at (4,4) {$\bullet$};\node at (5,5) {$\bullet$};\node at (6,4) {$\bullet$};\node at (7,3) {$\bullet$};\node at (8,2) {$\bullet$};\node at (9,1) {$\bullet$};\node at (10,0) {$\bullet$};

\node at (2,2.3){$p$};\node at (3,3.3){$p'$};
\node at (7,3.3){$q'$};\node at (8,2.3){$q$};
\node at (-.4,0){$q_{in}$};\node at (10.4,0){$q_{f}$};\node at (5,5.4){$r$};

\node[rotate=-45] at (1.4,3.7) {$p \earrow{a_i} (p', push(\gamma))$};
\node[rotate=45] at (8.3,3.7) {$(q', pop(\gamma)) \earrow{b_i} q$};

\path[draw,dashed] (1,-.3) -- (1,1);
\path[draw,dashed] (9,-.3) -- (9,1);
\path[draw,dashed] (2,-.3) -- (2,2);
\path[draw,dashed] (3,-.3) -- (3,3);
\path[draw,dashed] (8,-.3) -- (8,2);
\path[draw,dashed] (7,-.3) -- (7,3);
\path[draw,dashed] (4,-.3) -- (4,4);
\path[draw,dashed] (6,-.3) -- (6,4);
\path[draw,dashed] (5,-.3) -- (5,5);

\node at (5,-2) {Run in the \VPA $\mathcal{A}$ on $u$};

\begin{scope}[xshift=2cm]
\path[draw] (13,0) -- (13,.5);\path[draw,dotted] (13,.5) -- (13,1.5);
\path[draw] (13,1.5) -- (13,3.5);\path[draw,dotted] (13,3.5) -- (13,4.5);\path[draw] (13,4.5) -- (13,5);
\node at (13,5) {$\bullet$};\node at (12.4,5){$(r,r)$};
\node at (13,0) {$\bullet$};\node at (12.1,0){$(q_{in},q_f)$};
\node at (13,2) {$\bullet$};\node at (12.35,2){$(p,q)$};
\node at (13,3) {$\bullet$};\node at (12.35,3){$(p',q')$};
\node[rotate=90] at (13.35,2.5){\footnotesize$(a_i,b_i)$};
\node[rotate=90] at (13.35,0.75){\footnotesize$(v_1(1),I)\cdots$};
\node[rotate=90] at (13.35,4.25){\footnotesize$\cdots(a_h,b_h)$};
\node[rotate=90] at (13.35,-.55){\footnotesize$\widehat{u}=$};
\node at (13,-2) {\normalsize Run in the slicing automaton $\widehat{\mathcal{A}}$ on $\widehat{u}$};
\end{scope}

\end{tikzpicture}
\vlong{\caption{Slicing of a word $u \in \Lambda$ and evolution of the stack height for $u$.}\label{fig:mountain}}
\vshort{\caption{Slicing of a word $u \in \Lambda$.}\label{fig:mountain}}
\end{center}
\end{figure}

This construction will be later used  in Section~\ref{sec:general} for weighted languages.
In that case, we define the weight of a letter in $\widehat{u}$ by $|(a,b)|=|a| + |b|$, with the convention that $|I|=0$.
Moreover, we write $\widehat{\SQ}$ for the alphabet obtained similarly to $\widehat{\S}$ using $\SQ$ instead of $\Sl$. 
Note that the  slicing automaton $\widehat{\vpa}$ defined on $\widehat{\SQ}$ is $\widehat{\Sigma}$-closed and has $\widehat{\Sigma}$-diameter at most $2m^2$.
\begin{lemma}
\label{slicing_automaton}
If $\vpa$ is a \VPA{} accepting $L$, then
$\widehat{\vpa}$ is a finite automaton accepting 
$\widehat{L}=\{\widehat{u} : u\in L\cap\Lambda\}$.
\end{lemma}
\vlong{\begin{proof}
Because  transitions on push symbols do not depend on the top of the stack, 
transitions in $\widehat{\Delta}$  correspond to slices that
are valid for $\Delta$ (see Figure~\ref{fig:mountain}).
Finally, $\widehat{Q_{in}}$ ensures that a run for $L$ must start in $Q_{in}$
and end in $\Qf$, and $\widehat{\Qf}$ that a state at the top of the peak is  consistent from both sides.
\end{proof}}


\begin{proposition}\label{peak-diam}
Let $v\in \Lambda$ be s.t. $(p,q)\earrow{\widehat{v}} (p',q')$.
There is $w\in\Lambda$ s.t. $|w|\leq 2m^2$ and
$(p,q)\earrow{\widehat{w}} (p',q')$.
\end{proposition}

\subsection{Random Sketches}\label{picksketches}
We are now ready to build a tester for $L\cap \Lambda$. 
To test a word $u$ we use a property tester for the regular language $\widehat{L}$.
Regular languages are known to be $\eps$-testable for the Hamming distance with $\Order((\log 1/\eps)/\eps)$
non-adaptive queries on the input word~\cite{AKNS00}, that is queries that can all be made simultaneously.
Those queries define a small random sketch of $u_0$ that can be sent to the tester for approximating $R_{u_0}$.
Since the Hamming distance is larger than the edit distance, those testers are also valid for the latter distance.
Observe also that, for $u,v\in\Lambda_Q$, we have $\bdist(u,v)\leq 2\dist(\widehat{u},\widehat{v})$.
The only remaining difficulty is to provide to the tester an appropriate sampling on $\widehat{u}$ while processing $u$. 

We will proceed similarly for the general case in Section~\ref{sec:general}, but then we will have to consider weighted words.
Therefore we show how to sketch $u_0$ in that general case already.
Indeed, the tester of~\cite{AKNS00} was simplified for the edit distance in~\cite{NdioneLN13}, 
and later on adapted for weighted words in~\cite{n15}. 
We consider here an alternative approach that we believe simpler,  but slightly less efficient than the tester of~\cite{n15}.
\vlong{In particular, we introduce in Appendix~\ref{app:wregular} a new criterion, $\kappa$-saturation, that permits to significantly simplify the correctness proof of the tester compared to the one in~\cite{AKNS00} and in~\cite{n15}.}

Our tester for weighted regular languages is based on $k$-factor sampling on 
$\widehat{u}$ that we will simulate by an over-sampling built from a letter sampling on $u$, 
that is according to the weights of the letters of $u$ only. This  new sampling can be easily performed given a stream of $u$ using a standard reservoir sampling. 
\begin{definition}\label{def:samplingwk}
For a weighted word $u\in\LQ$, denote by $\w_k(u)$  the sampling over subwords of $u$ constructed as follows (see Figure~\ref{fig:samplingwk}):
\begin{compactenum}[(1)]
\item Sample a factor $u[i,i+k]$ of $u$ with probability $|u(i)|/|u|$.
\item If u(i) is in the push sequence of $u$, let $u[j,j']$ be the matching pop sequence of $u[i,i+k]$, including the first $k$ neutral symbols after the last pop symbol, if any. Add $u[j'-2k,j']$ to the sample.\footnote{Some matching pops of $u[i,i+k]$ may be ignored.}
\end{compactenum}
\end{definition}

\begin{figure}
\begin{tikzpicture}[scale=1,transform shape,
sample/.style={very thick,red},
]
\begin{scope}
\draw[sample] (0,0)--(1,0)--(1.5,.5)--(2,.5)--(3,1.5)--(4,1.5)--(4.5,2)--(5,2);
\draw[dashed] (0,0) -- (0,-1.2);
\node (ui) at (0,-1.5) {$u(i)$};
\node (uik) at (5,-1.5) {$u(i+k)$};
\draw[<->,>=stealth'] (ui) -- (uik);
\node at (2.5,-1.8) {$k+1$};
\draw[dashed] (5,2) -- (5,-1.2);
\draw (5,2) -- (6,3) -- (7,3)--(8,2);
\draw[dashed] (5,2) -- (8,2);
\draw [dashed] (1,0) -- (8,0);
\end{scope}
\begin{scope}[xshift = -2cm]
\draw [dashed] (10,0) -- (15.5,0);
\draw (10,2) -- (10.5,1.5) -- (11.5,1.5); 
\draw[sample] (11.5,1.5) -- (12,1.5) -- (12.5,1) -- (13,1) -- (13.5,.5) --(13.75,.5);
\draw[sample,dotted] (13.75,.5) -- (14.75,.5);
\draw[sample] (14.75,.5) -- (15,.5) -- (15.5,0) -- (15.75,0);
\draw[sample,dotted] (15.75,0) -- (16.75,0);
\draw[sample] (16.75,0)--(17,0);
\draw (17,0) -- (17.5,0);
\draw (17.5,0) -- (18,-.5);
\draw[dashed] (10,2) -- (10,-1.2);
\draw[dashed] (17,0) -- (17,-1.2);
\node (uj) at (10,-1.5) {$u(j)$};
\node (ujp) at (17,-1.5) {$u(j')$};
\draw[dashed] (11.5,1.5) -- (11.5,-1.2);
\node (ujpm) at (11.5,-1.5) {$u(j'-2k)$};
\draw[<->,>=stealth'] (ujp) -- (ujpm);
\node at (14.5,-1.8) {$2k+1$};
\draw[<->,>=stealth'] (15.5,-.75) -- (17,-.75);
\node at (16.25,-1) {$k$};
\end{scope}
\end{tikzpicture}
\caption{The sampling $\w_k(u)$ from Definition~\ref{def:samplingwk}: sample is in red, dotted parts are for omitted neutral symbols}\label{fig:samplingwk}
\end{figure}
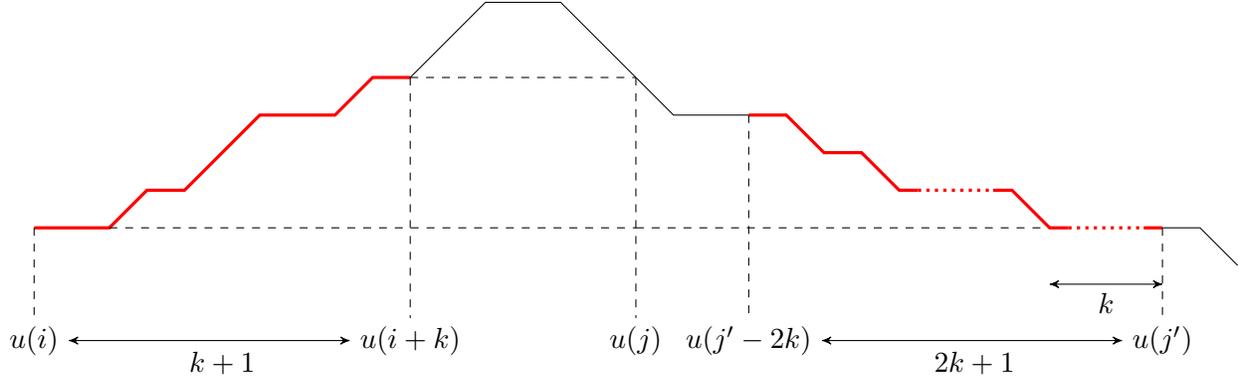

\begin{fact}\label{stream_sample}
There is a randomized streaming algorithm with memory $\Order(k + \log n)$ which, given $k$ and $u$ as input, samples $\w_k(u)$.
\end{fact}

\vlong{
\begin{proof}
(1) can easily be obtained using reservoir sampling. If the sampling enters the pop sequence as the current candidate is part of the push sequence, then (2) can be done for that candidate, and forgotten if the sampling eventually picks another one. That eventual candidate will not be part of the push sequence, so we are done.
\end{proof}
}

\begin{lemma}\label{leftsampling} 
Let $u$ be a weighted word, and let $k$ be such that $4k\leq |u|$.
Then $4k$ independent copies of $\w_k(u)$
over-sample the $k$-factor sampling on $\widehat{u}$.
\end{lemma}
\vlong{
\begin{proof}
Denote by $\widehat{\mathcal{W}}$ the $k$-factor sampling on $\widehat{u}$,
and by $\mathcal{W}$ some $4k$ independent copies of $\w_k(u)$.
For any $k$-factor $v$ of $\widehat{u}$, we will show that the probability that $\widehat{v}$ is sampled by 
$\widehat{\mathcal{W}}$ is at most the probability that  $\widehat{v}$ is a factor of an element sampled by $\mathcal{W}$.
For that, we distinguish the following three cases:
\begin{itemize}
\item $\widehat{v}$ contains only letters in $\{I\} \times \S_Q$. Then the probability that $\widehat{v}$ is sampled by $\widehat{\mathcal{W}}$ is equal to the probability that it is sampled by {$\w_k(u)$} in step (1).
\item $\widehat{v}$ starts by a letter $(a,b)$ in $\S_+\times \S_-$ or by a letter in $\S_Q\times \{I\}$. Then the probability that the $u(i)$ selected by $\w_k(u)$ is $a$ is at least half of the probability that {$\w_k(u)$} samples $\widehat{v}$, as a (push,pop) pair in $\widehat{u}$ has weight $2$ while a push has weight $1$ in $u$. Because $\widehat{v}$ is a $k$-factor, it is contained in $(u[i,i+k],u[j'-2k,j'])$. Hence, the probability that $\widehat{v}$ is sampled by {$\widehat{\mathcal{W}}$} is at most the probability that $\widehat{v}$ is a factor of an element sampled by $\w_k(u)$ in step (2).
\item $\widehat{v}$ starts by a letter in $\{I\} \times \S_Q$ but also contains letters outside of this set. 
Since $|\widehat{u}|\geq |u|/2$, we get
\[\Pr({\w_k(u)}\text{ samples }\widehat{v}) \geq 1/|u|
\quad\text{and}\quad
\Pr(\widehat{\mathcal{W}}\text{ samples }\widehat{v}){\leq}k/|\widehat{u}|\leq 2k/|u|.
\]
Thus the probability that one of the $4k$ samples of $\mathcal{W}$ 
has the factor $\widehat{v}$ is at least $1-(1-1/|u|)^{4k}$. 
As $1-(1-1/|u|)^{4k}\geq 1-\frac{1}{1+4k/|u|}=\frac{4k}{|u|+4k}\geq 2k/|u|$ when $|u|\geq 4k$, we conclude again that the probability that $\widehat{v}$ is sampled by $\widehat{\mathcal{W}}$ is at most the probability that $\widehat{v}$ is a factor of an element sampled by $\w_k(u)$ in step (2).
\end{itemize}
\end{proof}
}

We can now give an analogue of the property tester for weighted regular languages in $L\cap \La_Q$. 
For that, we use the following notion of approximation.
\begin{definition}\label{approx-regl}
Let $R\subseteq Q^2$. 
Then $R$ {\emph{$(\eps,\S)$-approximates} a balanced word $u\in (\Sc\cup\Sr\cup\S_Q)^*$} on ${\mathcal A}$,
if for all $p,q\in Q$:
(1) $(p,q)\in R$ when $p\earrow{u} q$;
(2) $u$ is $(\eps,\S)$-close to some word $v$ satisfying  $p\earrow{v} q$ when $(p,q)\in R$.
\end{definition}

Our tester is going to be robust enough in order to consider samples that do not exactly match the peaks we want to compress.
\begin{theorem}
\label{finalLtester}
Let $\vpa$ be a \VPA{} 
with $m\geq 2$ states and {$\S$-diameter} $d\geq 2$.
Let $\eps>0$, $\eta > 0$, $t=2\lceil 4dm^3(\log 1/\eta)/\eps\rceil$, $k=\lceil  4dm/\eps \rceil$
and $T=4kt$.
There is an algorithm that, given
$T$ random subwords $z_1,\ldots,z_{T}$ of some weighted word $v\in\Lambda_Q$, 
such that each $z_i$ comes from an independent sampling $\w_k(v)$,
outputs a set $R\subseteq Q\times Q$
that {$(\eps,\Sigma)$}-approximates $v$ on $\mathcal{A}$ with bounded error $\eta$.
\\
Let $v'$ be obtained from $v$ by at most $\eps |v|$ balanced deletions.
Then, the conclusion is still true if the algorithm is given an independent $\w_k(v')$ 
for each $z_i$ instead,
except that $R$ now provides a {$(3\eps,\S)$}-approximation. Last, each sampling can be replaced  by an over-sampling.
\end{theorem}
\vlong{\begin{proof}
The argument 
uses  as a subroutine the algorithm of Theorem~\ref{TheoRLtester}
for $\widehat{\mathcal{A}}$, where $\mathcal{A}$ has been extended to $\SQ$. 
Recall that $\mathcal{A}$ is $\S$-closed and its $\S$-diameter is also the $\widehat{\S}$-diameter of $\widehat{\mathcal{A}}$.  
Also observe that $\bdist_\Sigma(u,v)\leq 2\dist_{\widehat{\Sigma}}(\widehat{u},\widehat{v})$.

By Lemma~\ref{leftsampling}, the $T$ independent samplings $\w_k(v)$
provide us the sampling we need for Theorem~\ref{TheoRLtester}. 

For the case where we do not have an exact $k$-factor sampling on $v$ however,
we need to compensate for the prefix of $v$ of size $\eps |v|$ that may not be included in the sampling. This introduces potentially an additional error of weight $2\eps|v|$ on the approximation $R$.
\end{proof}} 

As a consequence we get our first streaming tester for  $L\cap\Lambda$.
\begin{theorem}
\label{TheoAlgoMountain}
Let $\vpa$ be a \VPA{} for $L$ with $m\geq 2$ states, and let $\eps,\eta > 0$.
Then there is a streaming $\eps$-tester for $L\cap\Lambda$
with {one-sided error $\eta$} and memory space ${\Order}((m^8\log (1/\eta)/\eps^2)(m^3/\eps+\log n))$,
where $n$ is the input length. 
\end{theorem}
\begin{proof}
We use Algorithm~\ref{AlgoExact} where we replace the current factor $u_0$ by $T=4kt$ independent samplings ${\mathcal W}_k(u_0)$. We know that such samplings can be computed using memory space $\Order(k + \log n)$ by Fact~\ref{stream_sample}.

By Proposition~\ref{peak-diam}, the slicing automaton has $\widehat{\Sigma}$-diameter $d$ at most $2m^2$.
Therefore, from Theorem~\ref{finalLtester}, taking $t={4}\lceil 4dm^3(\log 1/\eta)/\eps\rceil$ and $k=\lceil  4dm/\eps \rceil$ leads to the desired conclusion.
\end{proof}

%


\section{Algorithm With Sketching}\label{sec:general}
\subsection{Sketching Using Suffix Samplings}\label{sec:sampling}

We now describe the sketches used by our main algorithm. 
They are based on the generalization of the random sketches described in Section~\ref{picksketches}.
Moreover, they rely on 
a notion of suffix samplings, that ensures a good letter sampling on each suffix of a data stream.
Recall that the letter sampling on a weighted word $u$ samples a random letter $u(i)$ (with its position) with probability  $|u(i)|/|u|$.

\begin{definition}\label{suffixsampling}
Let $u$ be a weighted word and let $\alpha>1$.
An {\em $\alpha$-suffix decomposition of $u$ of size $s$} (see Figure~\ref{fig:suffixdecomposition}) is a 
sequence of suffixes $(u^l)_{1\leq l\leq s}$ of $u$ such that: 
$u^1=u$, $u^s$ is the last letter of $u$, and for all $l$, $u^{l+1}$ is a strict suffix of $u^l$ and
if $|u^{l}| > \alpha |u^{l+1}|$ then $u^{l}=a \cdot u^{l+1}$ where $a$ is a single letter. 

An {\em $(\alpha,t)$-suffix sampling on $u$ of size $s$} 
is an $\alpha$-suffix decomposition of $u$ of size $s$ with
$t$ letter samplings on each suffix of the decomposition.
\end{definition}

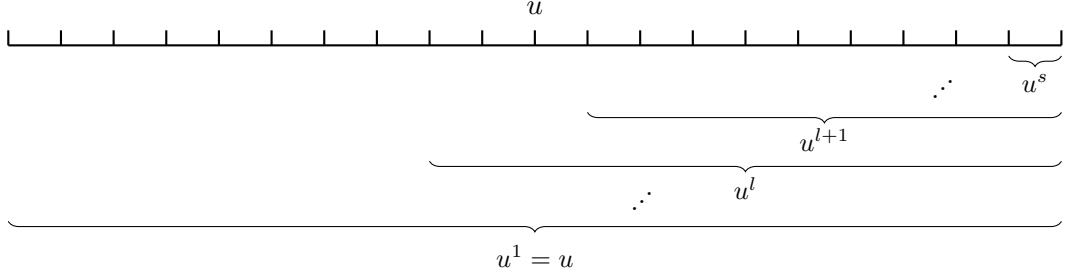
\begin{figure}
\begin{center}
\begin{tikzpicture}[scale=1,transform shape,
stack/.style={very thick,blue},
topstack/.style={very thick,red},
u/.style={very thick,OliveGreen},
a/.style={very thick,orange},
stackline/.style={thick,blue},
topstackline/.style={thick,red},
uline/.style={thick,OliveGreen},
aline/.style={thick,orange},
]
\begin{scope}
\draw [thick] (0,0) -- (14,0);
\foreach \i in {0,...,20}
{
\draw [thick] (.7*\i,0) -- (.7*\i,.2);
}
\node at (7,.5) {$u$};
\draw [decorate,decoration={brace,mirror,raise= 4pt,amplitude=4pt}] (13.3,0) -- (14,0);\node at (13.65,-.5){\small $u^s$};

\node[rotate=-45] at (12.5,-.5) {\small $\vdots$};
\draw [decorate,decoration={brace,mirror,raise= 4pt,amplitude=4pt}] (7.7,-.75) -- (14,-.75);\node at (10.85,-1.25){\small $u^{l+1}$};
\draw [decorate,decoration={brace,mirror,raise= 4pt,amplitude=4pt}] (5.6,-1.4) -- (14,-1.4);\node at (9.8,-1.9){\small $u^{l}$};
\node[rotate=-45] at (8.5,-2) {\small $\vdots$};

\draw [decorate,decoration={brace,mirror,raise= 4pt,amplitude=4pt}] (0,-2.2) -- (14,-2.2);\node at (7,-2.8){\small $u^1=u$};

\end{scope}
\end{tikzpicture}
\caption{An $\alpha$ suffix decomposition of $u$ of size $s$. For every $l$ either $|u^l|\leq\alpha|u^{l+1}|$ or $u^l = a\cdot u^{l+1}$ where $a$ is a letter.}\label{fig:suffixdecomposition}
\end{center}
\end{figure}

An $(\alpha,t)$-suffix sampling can be either concatenated to another one, or compressed
as stated below.
\begin{proposition}\label{concat}\label{simplify}
Given as input an $(\alpha,t)$-suffix sampling $D_u$ on $u$ of size $s_u$ and another one $D_v$ on $v$ of size $s_v$, 
there is an algorithm $\textbf{Concatenate}(D_u,D_v)$ computing an $(\alpha,t)$-suffix sampling on the concatenated word $u\cdot v$ 
of size at most $s_u+s_v$ in time $\Order(s_u)$.\\
Moreover, given as input an $(\alpha,t)$-suffix sampling $D_u$ on $u$ of size $s_u$,
there is also an algorithm $\textbf{Simplify}(D_u)$ computing an $(\alpha,t)$-suffix sampling on $u$ 
of size at most $2\lceil \log |u| /\log \alpha \rceil$ in time $\Order(s_u)$.
\end{proposition}
\begin{proof}
\vlong{We sketch those procedures. They are fully described in Algorithm~\ref{kSamplingAlgo}.}
For  $\textbf{Concatenate}$, it suffices to do the following. For each suffix $u^l$ of $D_u$: 
\begin{inparaenum}[(1)]
\item replace $u^l$ by $u^l\cdot v$; and 
\item replace the $i$-th sampling of $u^l$ by the $i$-th sampling of $v$ with probability $|v|/(|u|+|v|)$, for $i=1,\ldots,t$. 
\end{inparaenum}

For  \textbf{Simplify}, do the following.
For each suffix $u^l$ of $D_u$, from $l=s_u$ (the smallest one) to $l=1$ (the largest one): 
\begin{inparaenum}[(1)]
\item replace all suffixes $u^{l-1},u^{l-2},\ldots,u^m$ by 
the largest suffix $u^m$ such that $|u^m|\leq \alpha |u^l|$; and
\item suppress all samples from deleted suffixes.
\end{inparaenum}
\end{proof}

Using this proposition, one can easily design a streaming algorithm
constructing online a suffix decomposition of polylogarithmic size. 
Starting with an empty suffix-sampling $S$, simply concatenate $S$ with the next processed letter $a$ of the stream, and then simplify it.
\vlong{We formalize this, together with functions $\textbf{Concatenate}$ and $\textbf{Simplify}$,
 in Algorithm~\ref{kSamplingAlgo}}
\vlong{\begin{lemma}
\label{sampling_lemma}
Given a weighted word $u$ as a data stream and a parameter $\alpha>1$,
$\textbf{Online-Suffix-Sampling}$ in
Algorithm~\ref{kSamplingAlgo} constructs an $\alpha$-suffix sampling on $u$ of size at most $1+2\lceil \log |u| /\log \alpha \rceil$.
\end{lemma}
One can then slightly modify Algorithm~\ref{kSamplingAlgo} so that within each suffix of the decomposition it simulates $t$ letter samplings in order to construct an $(\alpha,t)$-suffix sampling. 
}

\iflong
\ifdraft \color{red} \fi
\begin{lstlisting}[caption={$\alpha$-Suffix Sampling},label=kSamplingAlgo,captionpos=t,abovecaptionskip=-\medskipamount,mathescape]
$\text{\bf Data structure:}$ 
  // $D$, $D_u$, $D_v$, $D_{\temp}$ stacks of items $(\sigma,b)$, one for each suffix 
  // of the decomposition where $\sigma$ encodes the weight and $b$ the $t$ samples    
$\text{\bf Code:}$    
$\textbf{Concatenate}(D_u,D_v)$
  $D\gets D_u$
  $(c_1,\ldots,c_t)\gets$ all $t$ samples on $v$ (the largest suffix in $D_v$)
  For each $(\sigma,b) \in S$ where $b=(b_1,\ldots,b_t)$
    Replace each $b_i$ by $c_i$ with probability $|v|/(|v|+\sigma)$(*@\label{sample}@*)
    Replace $(\sigma,b)$ by $(\sigma+|v|,{b})$(*@\label{add}@*)
  Append $D_v$ to the top of $D$(*@\label{push}@*)
  Return $D$
$\textbf{Simplify}(D_u)$
  $D\gets D_u$
  For each $(\sigma,b)\in D$ from top to bottom (*@\label{remove}@*)
    $D_{\temp}\gets$ elements $(\tau,c)\in D$ below $(\sigma,b)$ with $\tau\leq \alpha\sigma$ (*@\label{remove:T}@*)
    Replace $D_{\temp}$ in $D$ by the bottom most element of $D_{\temp}$
  Return $D$
$\textbf{Online-Suffix-Sampling}$ 
  $D\gets \emptyset$
  While $u$ not finished
    $a\gets\Next(u)$
    $\textbf{Concatenate}(D,a)$ where $a$ encodes the suffix sampling $(|a|,(a,\ldots,a))$
    $\textbf{Simplify}(D)$
  Return $D$    
\end{lstlisting}
\ifdraft \color{black} \fi
\fi

\subsection{The Algorithm}\label{sec:algo_sketch}

Our final algorithm is a modification of Algorithm~\ref{AlgoExact}: in particular it will approximate relations $R_v$ (in the spirit of Definition~\ref{approx-regl}), instead of exactly computing them. Therefore, it may fail at various steps and produce relations that do not correspond to any word. 
But still, it will produce relations $R$ such that for any $(p,q)\in R$, there is a  balanced word $u\in\S^*$ with $p\earrow{u} q$, that is $R\in\SQ$.

To mimic Algorithm~\ref{AlgoExact} we need to encode (compactly) each unfinished peak $v$ of the stack and $u_0$: for that we use the data structure described in Algorithm~\ref{sketchStructure}. Our final algorithm, Algorithm~\ref{AlgoFinal}, is simply Algorithm~\ref{AlgoExact} with this new data structure and corresponding adapted operations,
where $\eps' = \eps/ (6 \log n)$.


\begin{lstlisting}[caption={Sketch for an unfinished peak},label=sketchStructure,captionpos=t,abovecaptionskip=-\medskipamount,mathescape]
$\text{\bf Parameters:}$ real $\eps' > 0$, integer $T\geq 1$
$\text{\bf Data structure}$ for a weighted word $v\in\Prefix(\LQ)$
  Weights of $v$ and of its first letter $v(1)$
  Height of $v(1)$
  Boolean indicating whether $v$ contains a pop symbol 
  $(1+\eps')$-suffix decomposition $v^1,\dots,v^{s}$ of $v$ encoded by
     Estimates $|v^l|_{\low}$ and $|v^l|_{\high}$ of $|v^l|$
     $T$ independent samplings $S_{v^l}$ on $v^l$ // see details below
        with corresponding weights and heights
\end{lstlisting}

We now detail the methods, where we implicitly assume that each letter processed by the algorithm comes with its respective height and (exact or approximate) weight. They use functions $\textbf{Concatenate}$ and $\textbf{Simplify}$
described in~Proposition~\ref{concat}\vlong{  (and in details in Algorithm~\ref{kSamplingAlgo})}, while adapting them.
\begin{lstlisting}[caption={Adaptation of~Algorithm~\ref{AlgoExact} using sketches},label=AlgoFinal,captionpos=t,abovecaptionskip=-\medskipamount,mathescape] 
$\text{\bf  Run Algorithm~\ref{AlgoExact} using data structure from Algorithm~\ref{sketchStructure} and with the following adaptations:}$
$\text{\bf Adaption of functions from Proposition~\ref{concat}}$
 $\textbf{Concatenate}(D_u,D_v)$ with an exact estimate of $|v|$ is modified s.t.
   the replacement probability is now $|v|/(|u|_\high+|v|)$
   and $|u^l \cdot v|_z \gets |u^l|_z+|v|$, for $z=\low,\high$ 
 $\textbf{Simplify}(D_u)$ with $\alpha=1 + \eps'$ has now the relaxed condition $|u^m|_\high\leq (1+\eps') |u^l|_\low$
$\text{\bf Adaption of operations on factors used in Algorithm~\ref{AlgoExact}}$
 $\text{\bf Compute relation:}$ $R_v$
   Run the algorithm of Theorem(*@~\ref{finalLtester}@*) using samples in $D_v$ 
 $\text{\bf Decomposition:}$ $v_1\cdot v_2 \gets v$
   Find largest suffix $v^i$ in $D_v$ s.t. $v^i \in \Prefix(\Linf)$ // i.e. s.t. $v^i$ is in $v_2$
   $D_{v|v_1}\gets$ suffixes $(v^l)_{l<i}$ with their samples
   $D_{v_2}\gets$ suffix $v^i$ with its samples and weight estimates:  // for computing $R_{v_2}$
     - $(|v^{i}|_{\high}, |v^i|_{\low})$ when $v^{i-1}$ and $v^i$ differ by exactly one letter (then $v^i = v_2$)
     - $(|v^{i-1}|_{\high}, |v^i|_{\low})$ otherwise
 $\text{\bf Test:}$ $|u_0| \geq |v_2|/2$ using $|v_2|_\low$ instead of $|v_2|$
 $\text{\bf Concatenation:}$ $u_0 \gets (v_1\cdot R_{v_2}) \cdot u_0$
   $D_{v'}\gets (D_{v|v_1},R_{v_2})$ replacing each samples of $D_{v|v_1}$ in $v_2$ by $R_{v_2}$
   \\ The height of a sample determines whether it is in $v_2$
   $D_{u_0}\gets \textbf{Simplify}(\textbf{Concatenate}(D_{v'},D_{u_0}))$ 
\end{lstlisting}

In the next section, we show that the samplings $S_{v^l}$ are close enough to an $(1+\eps')$-suffix sampling on $v^l$.
This let us build an over-sampling of an $(1+\eps')$-suffix sampling. 
We also show that it only requires a polylogarithmic number of samples.
Then, we explain how to recursively apply the tester from Theorem~\ref{finalLtester} (with $\eps'$) in order 
to obtain the compressions at line~\ref{AE_R1} and~\ref{AE_R2} while keeping a cumulative error below $\eps$.
We now state our main result whose proof relies on Lemmas~\ref{stability} and~\ref{robustness}.
\begin{theorem}
\label{TheoMain}
Let $\vpa$ be a \VPA{} for $L$ with $m\geq 2$ states, and let $\eps,\eta >0$.
Then there is an $\eps$-streaming algorithm for $L$ with one-sided error $\eta$
and memory space 
$\Order(m^5 2^{3m^2}(\log^6 n)(\log 1/\eta)/ \eps^4)$,  
where $n$ is the input length.
\end{theorem}
\vlong{\begin{proof}
We use Algorithm~\ref{AlgoFinal}, which uses the tester from Theorem~\ref{finalLtester} for the compressions at lines~\ref{AE_R1} and~\ref{AE_R2} of Algorithm~\ref{AlgoExact}. We know from Lemma~\ref{RHeight} and Lemma~\ref{leftsampling} that it is enough to choose $\eps' = \eps/(6 \log n)$, 
$\eta'=\eta/n$, 
and Fact~\ref{diameter} gives us $d = 2^{m^2}$. 
Therefore we need $T= 2304m^42^{2m^2}(\log^2 n)(\log 1/\eta) /\eps^2$ independent 
$k$-factor samplings of $u$ augmented by one, with $k = 24m2^{m^2}(\log n)/\eps$.
Lemma~\ref{stability} tells us that using twice as many  samples from our algorithm,
that is for each $S_{v^l}$, is enough in order to over-sample them.

Because of the sampling variant we use, the size of each decomposition is at most $96 (\log^2 n) /\eps + \Order(\log n)$ by Lemma~\ref{stability}. The samplings in each element of the decomposition use memory space $k$, and there are $2T$ of them. Furthermore, each element of the stack has its own sketch, and the stack is of height at most $\log n$. Multiplying all those together gives us the upper bound on the memory space used by Algorithm~\ref{AlgoFinal}.
\end{proof}}

\subsection{Final Analysis}\label{sec:stab}\label{sec:rob}
As Algorithm~\ref{AlgoFinal} may fail at various steps, the relations it considers may not correspond to any word. 
However, each relation $R$ that it produces is still in $\S_Q$. Furthermore, 
the slicing automaton $\widehat{\mathcal{A}}$ that we define over $\widehat{\S_Q}$ is $\widehat{\S}$-closed. Fact~\ref{diameter} below bounds the $\widehat{\Sigma}$-diameter of $\widehat{\mathcal{A}}$ (which is equal to the $\S$-diameter of $\mathcal{A}$) by $2^{m^2}$. Note that for simpler languages, as those coming from a DTD, this bound can be lowered to $m$.

\begin{fact}\label{diam}
\label{diameter}
{Let $\mathcal{A}$ be a \VPA{} with $m$ states. Then
the $\S$-diameter of $\mathcal{A}$
is at most~$2^{m^2}$.} 
\end{fact}

\vlong{
\begin{proof}
A similar statement is well known for any context-free grammar given in Chomsky normal form.
Let $N$ be the number of non-terminal symbols used in the grammar.
If the grammar produces one {balanced} word from some non-terminal symbol, then it can also produce one whose length is at most $2^N$
from the same non-terminal symbol.
This is proved using a pumping argument on the derivation tree. We refer the reader to the textbook~\cite{HopcroftMU}. 

Now, in the setting of visibly pushdown languages one needs to transform $\mathcal{A}$ into a context-free grammar in Chomsky normal form. For that, consider first an intermediate grammar whose non-terminal symbols are all the $X_{pq}$ where $p$ and $q$ are states from $\mathcal{A}$: such a non-terminal symbol will produce exactly those words $u$ such that $p\earrow{u} q$, hence our initial symbol will be those of the form $X_{{q_0}{q_f}}$ where $q_0$ is an initial state and $q_f$ is a final state. The rewriting rules are the following ones:
\begin{itemize}
\item $X_{pp} \rightarrow \varepsilon$
\item $X_{pq} \rightarrow X_{pr}X_{rq}$ for any state $r$
\item $X_{pq} \rightarrow a X_{p'q'} b$ whenever one has in the automaton $p \earrow{a} (p', \push(\gamma))$ and $(q',\pop(\gamma)) \earrow{a} q$ for some push symbol $a$, pop symbol $b$ and stack letter $\gamma$.
\item $X_{pq} \rightarrow a X_{p'q}$ whenever one has in the automaton $p \earrow{a} p'$ for some neutral symbol $a$.
\item $X_{pq} \rightarrow X_{pq'}a$ whenever one has in the automaton $q' \earrow{a} q$ for some neutral symbol $a$.
\end{itemize}
Obviously, this grammar generates language $L(\mathcal{A})$. 

As we are here interested only in the length of the {balanced} words produced by the grammar, we can replace any terminal symbol by a dummy symbol $\sharp$. Now, once this is done we can put the grammar into Chomsky normal form by using an extra non-terminal symbol (call it $X_{\sharp}$ as it is used to produce the $\sharp$ terminal). As we have $m^2+1$ non-terminal in the resulting grammar we are almost done. To get to the tight bound announced in the statement, one simply removes the extra non-terminal symbol $X_{\sharp}$ and reasons on the length of the derivation directly.
\end{proof}
}

We first show that the decomposition, weights and sampling we maintain are close enough to an $(1+\eps')$-suffix sampling with the correct weights. Recall that $\eps' = \eps/ (6 \log n)$. 
\begin{lemma}[Stability lemma]
\label{stability}
Let $v,\mathcal{W}$ be an unfinished peak with a sampling maintained by the algorithm. 
Then $\mathcal{W}^{\otimes 2}$ over-samples an $(1+\eps')$-suffix sampling on $v$,
and $\mathcal{W}$ has size at most $144 (\log |v|) (\log n) /\eps + \Order(\log n)$.
\end{lemma}

%

\vlong{
Before proving the stability lemma, we first prove that Algorithm~\ref{AlgoFinal} maintains a strucutre that is not too far from $(1+\eps')$-suffix sampling.

\begin{proposition}
\label{weight_estimates}
Let $v$ be an unfinished peak, and let $v^1, \dots, v^s$ be the suffix decomposition maintained by the algorithm. The following is true: 
\begin{compactenum}[(1)]
\item $v^1,\dots,v^s$ is a valid $(1+\eps')$-suffix decomposition of $v$.
\item For each letter $a$ of every $v^l$, and for every sample $s$, $\Pr[S_{v^l} = a] \geq |a|/|v^l|_{\high}$.
\item Each $v^l$ satisfies $|v^l|_{\high} - |v^l|_{\low} \leq 2\eps'|v^l|_{\low}/3$.
\end{compactenum}
\end{proposition}
\begin{proof}
Property~(1) is guaranteed by the (modified) {\bf Simplify} function used in Algorithm~\ref{AlgoFinal}, which preserves even more suffixes than the original algorithm.

Properties (2) and (3) are proven by induction on the last letter read by Algorithm~\ref{AlgoFinal}. Both are true when no symbol has been read yet.

We start with property~(2). Let us first consider the case where we use bullet-concatenation after the last letter was read.
Then for all $v^l$, the (modified) {\bf Concatenate} function ensures $S_{v^l}$ becomes $a$  with probability $1/|v^l|_{\high}$. 
Otherwise, $S_{v^l}$ remains unchanged and by induction $S_{v^l}=b$  with probability at least $(1-1/|v^l|_{\high})|b|/(|v^l|_{\high}-1) = |b|/|v^l|_{\high}$, for each other letter $b$ of $v^l$.


The other case is that some $R_{v_2}$ is computed at line~\ref{AE_R2} of Algorithm~\ref{AlgoExact}. In this case, $v$ is equal to some $(v_1 \cdot R_{v_2}) \cdot u_0$ concatenation. 
For each suffix $(v_1 \cdot v_2)^l$ in $D_{(v_1 \cdot v_2)}$ containing $R_{v_2}$, we proceed in the same way with the {\bf Concatenate} function, replacing any sample in $v_2$ with $R_{v_2}$. Now consider $v_2^i$ the largest suffix of $D_{(v_1 \cdot v_2)}$ contained in $v_2$, and $v^l = R_{v_2} \cdot u_0$. We use the fact that {\bf Concatenate} looks at $|v^l|_{\high} \geq |u_0| + |R_{v_2}|$ for replacing samples. This means that we choose $R_{v_2}$ as a sample for $v^l$ with probability $(|v^l|_{\high}-|u_0|)/|v^l|_{\high} \geq |R_{v_2}|/|v^l|_{\high}$, and therefore the property is verified.

We now prove property~(3). If $v^l$ has just been created, it contains only one letter of weight $1$, and obviously $|v^l|_{\low} = |v^l|_{\high} = |v^l|$. 
In addition, unless some $R_{v_2}$ has been computed at line~\ref{AE_R2} of Algorithm~\ref{AlgoExact} when the last letter was read, then $|v^l|$ is only augmented by some exactly known $|a|$ or $|u_0|$ compared to the previous step. Therefore the difference $|v^l|_{\high} - |v^l|_{\low}$ does not change, and by induction it remains smaller than $2 \eps' |v^l|_{\low}/3$ which can only increase.
Now consider $R_{v_2}$ computed at line~\ref{AE_R2} and $v^l =  R_{v_2} \cdot u_0$. We again consider $v_2^i$ for the largest suffix in the decomposition of $v_1 \cdot v_2$ that is contained within $v_2$, as used in Algorithm~\ref{AlgoFinal}, and $v_2^{i-1}$ is the suffix immediately preceding $v_2^i$ in that decomposition. 

If $|v_2^{i-1}|_{\high} > (1+\eps') |v_2^i|_{\low}$, then from the {\bf Simplify} function, the difference between those two suffixes cannot be more than one letter, and then $v_2^i = v_2$. 
Therefore, we have $|R_{v_2} \cdot u_0|_{\high} = |v_2|_{\high} + |u_0|$ and $|R_{v_2} \cdot u_0|_{\low} = |v_2|_{\low} + |u_0|$. We conclude by induction on $|v_2|$. 

We end with the case $|v_2^{i-1}|_{\high} \leq (1+\eps') |v_2^i|_{\low}$.
By definition, $|R_{v_2} \cdot u_0|_{\high} = |v_2^{i-1}|_{\high} + |u_0|$ and $|R_{v_2} \cdot u_0|_{\low} = |v_2^i|_{\low} + |u_0|$.
Therefore the difference $|v^l|_{\high} - |v^l|_{\low}$ is at most $\eps'|v_2^i|_{\low}$.
Since the test at line~\ref{AE_test_left} of Algorithm~\ref{AlgoExact} (modified by ALgorithm~\ref{AlgoFinal}) was satisfied, we know that $|v_2^i|_{\low} \leq 2 |u_0|$,
and finally $\eps'|v_2^i|_{\low} \leq 2\eps'(|v_2^i|_{\low}+ |u_0|)/3 \leq 2 \eps'|v^l|_{\low}/3$, which concludes the proof.
\end{proof}

We can now prove the stability lemma.

\begin{proof}[Proof of Lemma~\ref{stability}]
The first property is a direct consequence of property (1) and (2) in Proposition~\ref{weight_estimates}, as in the proof of Lemma~\ref{leftsampling}.

The second is a consequence of the (modified) {\bf Simplify} used in Algorithm~\ref{AlgoFinal}: $D_{\temp}$ is defined as the set of suffixes below with $m<l$ such that $|v^m|_{\high}\leq (1+\eps')|v^l|_{\low}$. Because {\bf Simplify} deletes all but one elements from $D_{\temp}$, it follows that $|v^{l-2}|_{\high}> (1+\eps')|v^l|_{\low}$. Now, from property (3) of Proposition~\ref{weight_estimates} we have that $|v^l|_{\low}\geq |v^l|_{\high}-2\eps'|v^l|_{\low}/3\geq (1-2\eps'/3)|v^l|_{\high}$. Therefore we have that $|v^{l-2}|_{\high}> (1+\eps')(1-2\eps'/3)|v^l|_{\high}$

By successive applications, we obtain $|v^{l-6}|_{\high}> (1+\eps')^3(1-2\eps'/3)^3|v^l|_{\high}$. Now, as $|v^l|_{\high}>|v^l|$ and  $|v^l|\geq |v^l|_{\low}\geq (1-2\eps'/3)|v^l|_{\high}$ we have:
$|v^{l-6}|/(1-2\eps'/3)> (1+\eps')^3(1-2\eps'/3)^3|v^l|$. Equivalently, $|v^{l-6}|> (1+\eps')^3(1-2\eps'/3)^4|v^l|$.

Thus, the size of the suffix decomposition is at most 
$6 \log_{(1+\eps')^3(1-2\eps'/3)^4} |v| \leq 6 \log |v| /\log (1+\eps'/3 + \Order(\eps'^2)) \leq 144 (\log |v|) (\log n) /\eps + \Order(\log(n))$.
\end{proof}

}

Using the tester from Theorem~\ref{finalLtester}
for computing each $R$, we can then prove the robustness lemma.
\begin{lemma}[Robustness lemma]\label{robustness}
\label{RHeight}
Let $\vpa$ a \VPA{} recognizing $L$ and let $u \in \Sigma^n$. 
Let $R_\final$ be the final value of $R_\temp$ in the Algorithm~\ref{AlgoFinal}, using the tester from Theorem~\ref{finalLtester} at lines~\ref{AE_R1} and~\ref{AE_R2} of Algorithm~\ref{AlgoExact}.
If $u \in L$, then $R_\final \in L$; and if $R_\final \in L$, then $\bdist_\S(u,L) \leq \eps n$ with probability at least $1-\eta$.
\end{lemma}

\vlong{
\begin{proof}
One way is easy.
A direct inspection reveals that each substitution of a factor $w$ by a relation $R$
enlarges the set of possible $w$-transitions.
Therefore $R_\final\in L$ when $u\in L$.

For the other way, consider some word $u$ such that $R_\final \in L$.
Since the tester of Theorem~\ref{finalLtester} has bounded error $\eta' = \eta/n$
and was called at most than $n$ times, none of the calls fails with probability at least $1-\eta$. 
From now on we assume that we are in this situation. 

Let $h=\Depth(R_\final)$.
We will inductively construct sequences ${u_0}=u,\ldots,u_h=R_\final$ and ${v_h}=R_\final,\ldots,{v_0}$ such that for every $0 \leq l \leq h$, 
$u_l,{v_l} \in (\S_+ \cup \S_- \cup \S_{Q})^*$,
$\bdist_\S(u_l,{v_l}) \leq 3(h-l)\eps'|u_l|$ and $v_l\in L$.
{Furthermore, each word $u_l$ will be the word $u$ 
with some substitutions of factors by relations $R$ computed by the tester. 
Therefore, $\Depth(u_l)$ is well defined and will satisfy $\Depth(u_l)=l$.}
This will conclude the proof using that $\Depth(R_\final) \leq \log_{3/2} n$ 
from Lemma~\ref{NestedR}. This will give us $\bdist_\S(u,v_0) \leq 6\eps' n\log n \leq \eps n$.

\begin{figure}
\begin{center}
\begin{tikzpicture}[scale=0.75,transform shape,
u0/.style={very thick,red},
u0dot/.style={thick,red,dotted},
u1/.style={very thick,OliveGreen},
u2/.style={very thick,blue},
]
\begin{scope}[xscale=.7]
\draw [u0] (0,0) -- (2,2) -- (3,1) -- (6,4) -- (10,0) -- (13,3) -- (14.5,1.5) -- (15.5,2.5) -- (16.5,1.5) -- (20,5) -- (25,0);
\draw[u0dot] (1,1) -- (3,1);
\draw[u0dot] (11.5,1.5) -- (14.5,1.5);
\draw[u0dot] (14.5,1.5) -- (16.5,1.5);
\draw [u1] (0.2,-0.1) -- (1.2,.9) -- (3.2,.9) -- (6,3.7) -- (9.8,-.1) -- (10.2,-.1) -- (11.7,1.4) -- (16.7,1.4) -- (20,4.7) -- (24.8,-.1);
\draw [u2] (0.3,-0.25) -- (24.7,-.25);
\node[u2,xscale=1.4285] at (12.5,-1) {\large $R_\final$};
\node[u1,xscale=1.4285] at (2,.6) {\large $R$};
\node[u1,xscale=1.4285] at (13,1.1) {\large $R'$};
\node[u1,xscale=1.4285] at (15.5,1.1) {\large $R''$};
\end{scope}

\end{tikzpicture}
\vspace{-.2cm}
\caption{Constructing the words \textcolor{red}{$u_0$}, \textcolor{OliveGreen}{$u_1$} and \textcolor{blue}{$u_2$} as in Lemma~\ref{robustness} where $\Depth(R_\final)=2$}\label{fig:robustness}
\end{center}
\end{figure}
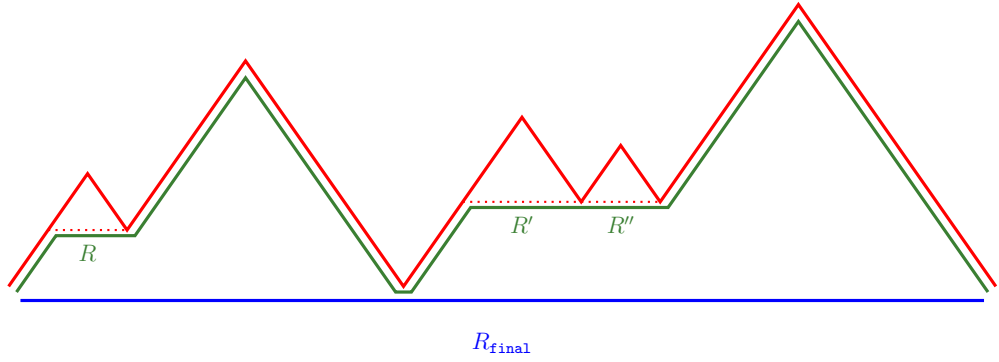

We first define the sequence $(u_l)_l$ {(see Figure~\ref{fig:robustness} for an illustration)}. Starting from $u_0=u$, let $u_{l+1}$ be the word $u_l$
where some factors in $\La_{Q}$ have been replaced by a $(3\eps',\S)$-approximation in $\S_{Q}$. {These correspond to all the approximations eventually performed by the algorithm that did not involve a symbol already in $\S_{Q}$. 
 Observe that after this collapse, the symbol is still a $(3\eps',\S)$-approximation.}
In particular, $u_h=R_\final$, $u_l\in (\S_+ \cup \S_- \cup \S_{Q})^*$  and $\Depth(u_l)=l$ by construction.

We now define the sequence $(v_l)_l$ such that $v_l\in L$.
Each letter of $v_l$ will be annotated by an accepting run of states for $\mathcal{A}$.
Set $v_h=R_\final$ with an accepting run from $p_{\mathit{in}}$ to $q_{\mathit{f}}$ for some $(p_{\mathit{in}},q_{\mathit{f}})\in R_\final\cap(\Qin\times\Qf)$.
Consider now some level $l<h$.
Then $v_l$ is simply $v_{l+1}$ where some letters $R\in\S_{Q}$ in common with $u_{l+1}$ are replaced by some factors in $w\in(\La_{Q})^*$ as explained in the next paragraph. 
Those letters are the ones that are present in $u_l$ but not $u_{l+1}$, and are still present in $v_{l+1}$ (i.e. they have not been further approximated down the chain from $u_{l+1}$ to $u_h$, or deleted by edit operations moving up from $v_h$ to $v_{l+1}$).

Let $w\in (\La_{Q})^*$ be one of those factors and $R\in \S_{Q}$ 
its respective $(3\eps',\S)$-approximation. 
By hypothesis $R$ is still in $v_{l+1}$ and corresponds to a transition $(p,q)$ of the accepting run of $v_{l+1}$. We replace $R$ by a factor $w'$ such that $p\earrow{w'} q$ and $\bdist_\S(w,w')\leq 3\eps' |w|$,
and annotate $w'$ accordingly.
By construction, the resulting word $v_{l}$ 
satisfies $v_l\in L$ and $\bdist_\S(u_{l},v_{l}) \leq 3(h-l)\eps' |u_l|$.
\end{proof}
}

\iflong
\else
\newpage
\fi
\bibliographystyle{plain}
\bibliography{ref}

\iflong
\newpage
\else
\end{document}